\documentclass[english,12pt]{article}

\pdfoutput=1

\usepackage[T1]{fontenc}
\usepackage[latin9]{inputenc}
\usepackage{amssymb}
\usepackage{babel}
\usepackage{amsthm}
\usepackage{verbatim}
\usepackage{amsrefs}
\usepackage{amsmath}
\usepackage{listings}
\usepackage{fullpage}
\usepackage{graphicx}
\usepackage{qtree}
\usepackage[usenames,dvipsnames]{color}
\usepackage{array}
\usepackage{enumerate}
\usepackage{polynom}
\usepackage{tikz}
\usepackage{fancyhdr}
\usepackage{etoolbox}
\usepackage{algorithm}
\usepackage[noend]{algpseudocode}

\newtheorem{prop}{Proposition}
\newtheorem{theorem}{Theorem}
\newtheorem{lemma}{Lemma}

\newcommand{\p}[1]{\left(#1\right)}
\newcommand{\st}[1]{\left\{#1\right\}}

\newcommand{\bk}[1]{\left[#1\right]}

\newcommand{\quot}[1]{``#1''}

\newcommand{\lb}{\hspace*{\fill}}

\newcommand{\namehead}[3]{
\lstset{breaklines=true, morecomment=[l]{//}, frame=single, showstringspaces=false, numbers=left}
\begin{flushright}
Nathan Fox\\
#2\\
#3\\
\end{flushright}
\ifstrequal{#1}{.}{}{
\begin{center}
{\Large Homework #1}
\end{center}}
}
\usepackage{mathdots}
\usepackage{multicol}
\usepackage{calc}
\usepackage{tikz}
\usetikzlibrary{calc}
\usetikzlibrary{shapes}

\renewcommand{\bk}[1]{[#1]}

\renewcommand{\p}[1]{(#1)}
\newcommand{\pb}[1]{\left(#1\right)}

\newcommand{\T}[1]{T\p{#1}}

\newcommand{\ps}{\textsc{pspace}}

\newcommand{\ttt}{true}
\newcommand{\fff}{false}

\newcommand{\mathsc}[1]{\text{\textsc{#1}}}
\newcommand{\game}[1]{\textsc{\lowercase{#1}}}
\newcommand{\sgame}[2]{$\mathsc{\lowercase{#1}}_{#2}$}

\begin{document}

\title{On the Computational Complexities of Various Geography Variants}
\author{Nathan Fox\footnote{Mathematics and Statistics, Canisius College, Buffalo, New York,
\texttt{fox42@canisius.edu}}, Carson Geissler\footnote{The College of Wooster, Wooster, Ohio,
\texttt{cgeissler20@wooster.edu}
}}
\date{}

\maketitle

\abstract{Generalized Geography is a combinatorial game played on a directed graph. Players take turns moving a token from vertex to vertex, deleting a vertex after moving the token away from it. A player unable to move loses. It is well known that the computational complexity of determining which player should win from a given position of Generalized Geography is \ps-complete. We introduce several rule variants to Generalized Geography, and we explore the computational complexity of determining the winner of positions of many resulting games. Among our results is a proof that determining the winner of a game known in the literature as Undirected Partizan Geography is \ps-complete, even when restricted to being played on a bipartite graph.}

\section{Introduction}\label{s:intro}
The underpinning for our work is the game of Generalized Geography, also known as Generalized Vertex Geography, and, henceforth in this paper, simply as \game{Geography}. \game{Geography} is played on a directed graph $G$, with a token that begins on some vertex $v$ of $G$. Players take turns moving the token along arcs of $G$. After a player moves the token, that player deletes the vertex the token moved from, so that no later move can visit that vertex. A player unable to move loses. The name \quot{Geography} comes from the children's game of the same name in which players take turns naming cities. Cities cannot repeat, and the first letter in each city's name must be the same as the last letter of the previous city's name. The cities correspond to vertices in $G$, the letter matching restriction to arcs in $G$, and the \quot{no repeats} rule to the vertex deletion rule.

As described, \game{Geography} is an impartial combinatorial game with the normal play convention. Combinatorial games are finite games of perfect information played with two players. Under normal play, a player loses when unable to move. Finally, in an impartial game, the same set of moves is available in a given position regardless of which player's turn it is. In \game{Geography}, this amounts to the fact that there is a single token that players take turns moving. This is as opposed to partizan games, like chess, in which each player has their own set of moves available to them in a given position~\cite{winways}.

A natural question to ask about combinatorial games is, given a position, which player should win assuming both players play optimally? For some games, such as Nim, it is straightforward to decide which player wins~\cite{nim}. For other games, such as chess, it is far less clear~\cite{chess}. Computational complexity is a formal lens through which we can quantify the difficulty of determining which player should win. 
For brevity, we henceforth refer to this question simply as the \emph{complexity of the game}. 
A classical result states that 
\game{Geography} is \ps-complete~\cite{geogreduct}. 
On the other hand, since there is a polynomial time algorithm to determine the winner of Nim, 
Nim is in \textsc{p}. No problem in \textsc{p} is \ps-complete unless \textsc{p}=\ps, a unsolved proposition stronger than \textsc{p}=\textsc{np} that is generally expected to be false.

There is a long tradition of studying the 
complexities of variants of \game{Geography}~\cite{fsugeog, fsgeog, monti2018variants, bodlaender1993complexity, darmann2014shortest, schaefergeog}. In this paper, we examine four ways of modifying \game{Geography}, which can be combined to yield a multitude of \game{Geography} variants. Our motivation comes from the board game \textit{Santorini}~\cite{santorini}; our variants adjust the rules of \game{Geography} to become more like \textit{Santorini}. Here are the variations we consider:
\begin{description}
\item[Directed versus Undirected:] \game{Geography} is played on a directed graph, but it can just as easily be played on an undirected graph.
\item[Impartial versus Partizan:] As mentioned before, \game{Geography} is an impartial game. The game can be made partizan by replacing the single token with two tokens, one controlled by each player, with the additional rule that the tokens cannot occupy the same vertex simultaneously~\cite{fsgeog}. We refer to \game{Geography} with only this alteration as \game{partizan Geography}.
\item[Restricted versus Free Deletion:] In \game{Geography}, a move consists of moving a token from a vertex to a successive vertex and then deleting the vertex moved from. In this paper, we refer to this rule as \emph{restricted deletion}. We contrast this with a different move rule, which we call \emph{free deletion}. Under free deletion, a move consists of moving a token from a vertex to a successive vertex and then deleting a vertex that could have been moved from. In the undirected case, this amounts to deleting any neighbor of the vertex moved to; in the directed case it amounts to deleting any predecessor of the vertex being moved to.
\item[Stacked Vertices:] In this variant, each vertex has a positive integer height associated to it. When a vertex would be deleted, instead its height is decreased by $1$. If the height of a vertex becomes $0$, it is deleted. In all cases, we treat the maximum allowed height as a fixed parameter, thereby still guaranteeing that the length of the game is polynomial in the number of vertices in the graph (Proposition~\ref{prop:pspace}).

Without introducing any additional modifications, stacked vertices cannot actually change the complexity of a game, because a game with stacked vertices would be equivalent to a game without stacked vertices where a vertex of height $k$ is replaced by $k$ vertices (analogously to the proof of Proposition~\ref{prop:1vs2}). So, along with the setup of stacked vertices comes a moving restriction: From a vertex of height $k$, a move is only legal to a successive vertex of height at least $k-1$. Note that this move restriction only becomes relevant for maximum heights of $3$ or greater; any game with maximum height $2$ is equivalent to a game with maximum height~$1$ (Proposition~\ref{prop:1vs2}).
\end{description}
Our notation for a game consists of three letters and a numerical subscript. The first letter, \game{D} or \game{U}, stands for directed or undirected. The second letter, \game{I} or \game{P}, stands for impartial or partizan. The third letter, \game{R} or \game{F}, stands for restricted or free deletion. Finally, a numerical subscript denotes the maximum allowed height of a stacked vertex. A maximum height of $1$ is equivalent to not allowing stacked vertices, so we may omit the subscript when it would be $1$. 
Under this notation, \game{Geography} is denoted \game{DIR} or \sgame{DIR}{1}.

Table~\ref{t:summary} contains a summary of known results and our results. Of particular note are all of the games with restricted deletion and no stacking, as all have been considered previously. The game \game{UIR}, known in the source studying it as \game{Undirected Vertex Geography}, has a straightforward solution~\cite{fsugeog}, which extends to the stacked game \sgame{UIR}{2}. That same source mentions the game \game{UPR} near the end, stating only that it is \textsc{np}-hard. We prove in Section~\ref{ss:upr} that it is, in fact, \ps-complete. The game of \game{Partizan Geography} (\game{DPR}) 
is known to be \ps-complete, even when restricted to bipartite graphs with maximum in/out degree at most two and maximum total degree at most three~\cite{fsgeog}. 
The game of \game{Geography} (\game{DIR}) 
is also still \ps-complete under such restrictions. It remains \ps-complete if the graph is additionally required to be planar~\cite{geogreduct}.

The proofs in this paper are organized with the non-stacking {\ps} reductions first (Theorems~\ref{thm:dif} through~\ref{thm:upf}), starting with the simplest reduction. Then, we prove Theorem~\ref{thm:uif} about the game \game{UIF}, and we conclude with Theorem~\ref{thm:uir4} about \sgame{UIR}{4}.

\begin{table}
\begin{center}
\begin{tabular}{|c|c|c|}\hline
\textbf{Game} & \textbf{Complexity} & \textbf{Analysis}\\\hline\hline
\game{DIR} & \ps-complete, even for bipartite with degree $\leq 3$ & \cite{schaefergeog,geogreduct}\\\hline
\game{DIF} & \ps-complete, even for bipartite graphs & Section~\ref{ss:dif}\\\hline
\game{DPR} & \ps-complete, even for bipartite with degree $\leq 3$ & \cite{fsgeog}
\\\hline
\game{DPF} & \ps-complete, even for bipartite graphs & Section~\ref{ss:dpf}\\\hline
\sgame{UIR}{2} & In \textsc{p} & \cite{fsugeog}
\\\hline
\sgame{UIR}{4} & \ps-complete, even for bipartite with degree $\leq 3$ & Section~\ref{s:uik}\\\hline
\sgame{UIF}{2} & In \textsc{p} for bipartite graphs & Section~\ref{s:ui1}\\\hline
\game{UPR} & \ps-complete, even for bipartite graphs & Section~\ref{ss:upr}\\\hline
\game{UPF} & \ps-complete & Section~\ref{ss:upf}\\\hline
\end{tabular}
\end{center}
\caption{Summary of computational complexities of \game{Geography} variants}
\label{t:summary}
\end{table}

\section{Preliminaries}\label{s:prelim}

\subsection{Notation and Conventions}
When we refer to a \emph{game}, we refer to all possible instances or positions of that game. All of our games have two players, whom we conventionally name Left and Right. In our proofs, we use the pronouns \emph{she} for Left and \emph{he} for Right to help distinguish between the players. In a particular position, either player might be the active player, though, following convention, positions we construct always have Left as the active player, if possible. Furthermore, if a game is impartial, a position is called an N-position if the next player to move should win and a P-position otherwise.

Sometimes, we want to refer to multiple games at once. In place of \game{D} or \game{U}, we write \game{X} when it can be either. Similarly, we can write \game{Y} in place of \game{I} or \game{P} and \game{Z} in place of \game{R} or \game{F}.  
The game \sgame{XYZ}{k} has an underlying graph $G$ that could either be directed or undirected. If $G$ is undirected, given vertices $v$ and $w$ in $G$, we write $v\sim w$ if $v$ and $w$ are adjacent. If $G$ is directed, given vertices $v$ and $w$ in $G$, we write $v\to w$ if there is an arc from $v$ to $w$. In this case, we say that $w$ \emph{succeeds} $v$ and that $v$ \emph{precedes} $w$. 
A move in \sgame{XYZ}{k} consists of the active player moving a token from some vertex $v$ to another vertex $w$ and then reducing the height of a vertex $u$ by $1$ (and deleting $u$ if its height reaches $0$). For convenience of notation, we encode such a move as $\pb{v,w,u}$. If $u=v$ (as must be the case under restricted deletion), we say that the move is \emph{regular}; otherwise it is \emph{irregular}.

This paper includes several diagrams of positions. Vertices of graphs are denoted by circles which may contain symbols. The symbol $T$ denotes the token in an impartial game; the symbols $L$ and $R$ denote Left's and Right's tokens in a partizan game. A number inside a circle denotes the height of a vertex in a game with stacking. Other symbols in a circle have no meaning within the game and are present only to clarify the diagram. See Figure~\ref{fig:2to1} on page~\pageref{fig:2to1} for an example of partizan positions with and without stacking, and see Figure~\ref{fig:dif} on page~\pageref{fig:dif} for an example of an impartial position without stacking.

In this paper, we are concerned with the computational complexity of a given game. All of the games \sgame{XYZ}{k} are in {\ps} (Proposition~\ref{prop:pspace}), so proving \ps-completeness amounts only to proving \ps-hardness. We are interested in separating games that are in \textsc{p} from games that are \ps-complete, so we employ polynomial-time reductions~\cite{karp}. We use Karp's notation $K_1\leq_p K_2$ to mean that game $K_1$ is polynomial-time-reducible to game $K_2$, which means that the complexity of $K_2$ is at least that of $K_1$. Several of the proofs in this paper proceed via reductions. All of our reductions take place in polynomial time, because in each case a position is being transformed into another position, the latter of which has size polynomial in the size of the former, and where each step of the transformation can be performed in polynomial time.

\subsection{Quantified Boolean Formulas}
In order to prove that games are \ps-complete, we need a \ps-complete problem to start reductions from. The classical choice is the decision problem \game{TQBF}, sometimes known as \game{QSAT}, which asks whether a given quantified boolean formula is true. Without loss of generality, we can always assume that the expression is given in the form $\exists x_1\forall x_2\exists x_3\cdots\exists x_{n-1}\forall x_n P\p{x_1,x_2,\ldots,x_n}$ for some predicate $P$. This problem is still \ps-complete if we require $P$ to be a 3-CNF formula (conjunctive normal form with exactly three literals per clause)~\cite{gareyjohnson}. For example, we might consider an expression such as $\exists x_1\forall x_2\exists x_3\forall x_4\pb{x_1\vee\bar{x}_2\vee x_3}\wedge\pb{x_1\vee\bar{x}_3\vee\bar{x}_4}$. Conventionally, we use $n$ to refer to the number of variables (always even) and $m$ to refer to the number of clauses. Henceforth in this paper, the term \game{TQBF} assumes this normal form with all mentioned restrictions in place. 

Note the distinction between \emph{variable} and \emph{literal}. A variable is one of the symbols $x_i$; a literal is a variable or its negation. We occasionally say that a literal \emph{contains} a variable.

\game{TQBF} can be formulated as a game as follows. Two players, Left and Right, are given a 3-CNF formula $P\pb{x_1,x_2,\ldots,x_n}$. The players, starting with Left, take turns setting the truth values of the variables in increasing order of the index. (Left starts by setting $x_1$, then Right sets $x_2$, etc.) Once all variables have been assigned truth values, Left wins if $P\pb{x_1,x_2,\ldots,x_n}$ is true, and Right wins if it is false. Deciding who wins this game is precisely the same as deciding the truth value of $\exists x_1\forall x_2\exists x_3\cdots\exists x_{n-1}\forall x_n P\p{x_1,x_2,\ldots,x_n}$. 

\subsection{Simple Propositions}
The following results provide some basic relationships between our games and a foundation for the more substantial results to follow in later sections.

\begin{prop}\label{prop:pspace}
The game \sgame{XYZ}{k} is in \ps.
\end{prop}
\begin{proof}
To represent a position in \sgame{XYZ}{k}, we need to encode the underlying graph $G$, the height of each vertex, the locations and owners of all tokens, and which player is to move. If $G$ has $n$ vertices, it can be encoded using $O\p{n^2}$ bits by listing its vertices and its arcs/edges. Since we treat $k$ as a constant, $O\p{n}$ bits are required to encode the heights of all of the vertices. Then, since each of our variants uses at most two tokens, the remainder of the information can be encoded using $O\p{\log n}$ bits, identifying by a numerical index each vertex with a token. Hence, a position with $n$ vertices can be encoded using $O\p{n^2}$ bits, a polynomial in $n$.

We must now show that the winner can be determined using space polynomial in $n$. Every move of the game must reduce the height of a vertex, so the total number of moves in the game is at most the total height of all of the vertices. This, in turn, is at most $nk$, as each vertex can have height at most $k$. So, we can explore the entire game tree with backtracking, needing to store at most $O\p{n}$ positions in memory simultaneously. Each position can be stored using $O\p{n^2}$ bits, so we need to store at most $O\p{n^3}$ bits of game information. To track the winner in each position, we need only an additional bit per position, so we can determine the winner using $O\p{n^3}$ space, a polynomial in $n$, as required.
\end{proof}

\begin{prop}\label{prop:redtriv}
If $k<\ell$, the game \sgame{XYZ}{k} is no more complex than \sgame{XYZ}{\ell}, even if the same restriction is placed on the underlying graphs for both.
\end{prop}
\begin{proof}
If $k<\ell$, any instance of \sgame{XYZ}{k} is also an instance of \sgame{XYZ}{\ell} with the same underlying graph.
\end{proof}

\begin{prop}\label{prop:udtriv}
The game \sgame{UYZ}{k} is no more complex than \sgame{DYZ}{k}.
\end{prop}
\begin{proof}
To convert an instance of \sgame{UYZ}{k} to an equivalent instance of \sgame{DYZ}{k}, replace each edge with a pair of arcs, one going each direction.
\end{proof}

\begin{prop}\label{prop:1vs2}
The complexities of the games \sgame{XYZ}{1} and \sgame{XYZ}{2} are the same, and they are still the same (as each other) when restricted to bipartite graphs.
\end{prop}
\begin{proof}
By Proposition~\ref{prop:redtriv}, \sgame{XYZ}{2} is no less complex than \sgame{XYZ}{1}. We now give a polynomial-time reduction from \sgame{XYZ}{2} to \sgame{XYZ}{1}, thereby also showing that \sgame{XYZ}{2} is no more complex than \sgame{XYZ}{1}. Additionally, under our reduction the underlying graph in one instance is bipartite if and only if the other is. Given an instance $K$ of \sgame{XYZ}{2} on a graph $G$, create an instance $\T{K}$ of \sgame{XYZ}{1} with underlying graph $G'$ as follows:
\begin{itemize}
\item For a vertex $v\in G$, $G'$ has a vertex $v_1$ if $v$ has height $1$ in $K$, and $G'$ has vertices $v_1$ and $v_2$ if $v$ has height $2$ in $K$.
\item If $G$ is undirected, if $v\sim w$ in $G$, then $v_1\sim w_1$ in $G'$. Furthermore, $v_1\sim w_2$, $v_2\sim w_1$, and $v_2\sim w_2$ whenever both vertices involved exist. Analogously, if $G$ is directed, if $v\to w$ in $G$, then $v_1\to w_1$ in $G'$. Furthermore, $v_1\to w_2$, $v_2\to w_1$, and $v_2\to w_2$ whenever both vertices involved exist.
\item If vertex $v$ of height $1$ has a token on it in $K$, then $v_1$ has that token on it in $\T{K}$.
\item If vertex $v$ of height $2$ has a token on it in $K$, then $v_2$ has that token on it in $\T{K}$.
\end{itemize}
See Figure~\ref{fig:2to1} for an example of this transformation. This example is directed and partizan, but an analogous transformation works if the game is undirected and/or impartial.

\begin{figure}
\begin{center}
\begin{tikzpicture}[thick, node distance=1.5cm, circ/.style={draw, circle, minimum size=25pt}]

\node[circ] at (0, 0) (0) {$2$};
\node[circ] at (2, 0)  (1) {$2L$};
\node[circ] at (4, 0)  (2) {$1$};
\node[circ] at (6, 0) (3)  {$1R$};
\node[circ] at (8, 0) (4)  {$2$};

\draw [->] (0)->(1);
\draw [->] (1)->(2);
\draw [->] (2)->(3);
\draw [->] (3)->(4);
\end{tikzpicture}\\\lb

{\Huge $\Downarrow$}\\

\begin{tikzpicture}[thick, node distance=1.5cm, circ/.style={draw, circle, minimum size=25pt}]

\node[circ] at (4, 0) (2) {$$};
\node[circ] at (6, 0)  (3) {$R$};
\node[circ] at (8, 1)  (4a) {$$};
\node[circ] at (8, -1)  (4b) {$$};
\node[circ] at (2, 1) (1a) {$$};
\node[circ] at (2, -1) (1b)  {$L$};
\node[circ] at (0, 1) (0a) {$$};
\node[circ] at (0, -1) (0b)  {$$};

\draw [->] (0a)->(1a);
\draw [->] (0a)->(1b);
\draw [->] (0b)->(1a);
\draw [->] (0b)->(1b);
\draw [->] (1a)->(2);
\draw [->] (1b)->(2);
\draw [->] (2)->(3);
\draw [->] (3)->(4a);
\draw [->] (3)->(4b);
\end{tikzpicture}
\end{center}
\caption{Reduction from \sgame{DPZ}{2} to \sgame{DPZ}{1}.}
\label{fig:2to1}
\end{figure}
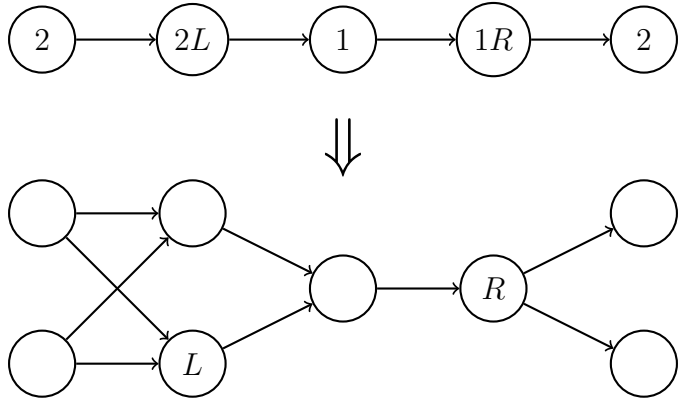

We claim that Left wins $\T{K}$ if and only if she wins $K$. We argue by induction on the total height of all vertices in $K$. As a base case, note that if all vertices in $K$ have height $1$, then $K=\T{K}$ (aside from vertex subscripts in $\T{K}$), so they have the same winner. Now, suppose all games $\bar{K}$ with smaller total height than $K$ have the property that the same player wins $\bar{K}$ and $\T{\bar{K}}$. Now, consider an optimal move $\pb{v,w,u}$ for the active player in $K$. Denote by $K_1$ the resulting instance of \sgame{XYZ}{2}. In $\T{K}$, any move of the form $\pb{v_i,w_j,u_\ell}$ is legal, provided all vertices exist and $i$ is as large as possible (so that there is actually a token there to be moved). Choose the move where $i$, $j$, and $\ell$ are all as large as possible. The result is an instance $K_1'$ of \sgame{XYZ}{1}. Note that $K_1'=\T{K_1}$, so $k_1$ and $K_1'$ have the same winner by induction. Hence, $K$ and $\T{K}$ have the same winner, as required.

Finally, we claim that the $G$ is bipartite if and only if $G'$ is bipartite. First, if $v,w,\ldots, z$ is an odd cycle in $G$, then $v_1,w_1,\ldots,z_1$ is an odd cycle in $G'$. On the other hand, if $\pb{A,B}$ is a bipartition of $G$, then $\pb{\st{v_i:v\in A},\st{w_i:w\in B}}$ is a bipartition of $G'$.
\end{proof}

\section{Directed Geography with Free Deletion}\label{s:dif}\label{s:dp}
In this section, we add free deletion to both \game{Geography} and \game{Partizan Geography}. We obtain that these games, \game{DIF} and \game{DPF}, are both \ps-complete, even when restricted to bipartite graphs.

\subsection{Complexity of DIF}\label{ss:dif}

We start by proving that \game{DIF} is \ps-complete on bipartite graphs. Our reduction is a variation of Lichtenstein and Sipser's classical reduction from \game{TQBF} to \game{Geography}, along with their alteration to make the resulting graph bipartite~\cite{geogreduct}. Those authors are also able to place degree restrictions on the vertices in the resulting graph, but their method for in-degree restriction does not work under free deletion.
\begin{theorem}\label{thm:dif}
The game \game{DIF} is \ps-complete, even when restricted to bipartite graphs.
\end{theorem}
\begin{proof}
We proceed via a reduction from \game{TQBF} to \game{DIF}. Let $Q=\exists x_1\forall x_2\cdots\exists x_{n-1}\forall x_n c_1\wedge c_2\wedge\cdots\wedge c_m$ be an instance of \game{TQBF}. We construct an instance of \game{DIF} as follows:
\begin{itemize}
\item For each variable $x_i$, there is a structure called a \emph{variable gadget}. If $i$ is even, the variable gadget is a diamond with a pendant. The top vertex precedes each of the left and right vertices, each of which precedes a vertex we call the \emph{joining vertex}, which precedes the pendant (bottom) vertex. If $i$ is odd, the variable gadget is a hexagon, with a left path and a right path each with two vertices, a top vertex preceding each of the left and right paths, and a joining (bottom) vertex succeeding each of the left and right paths.
See Figure~\ref{fig:difvar} for depictions of both types of variable gadget.
\item The token begins at the top vertex of the $x_1$ gadget.
\item For each clause $c_j$, there is a structure called a \emph{clause gadget}. It contains a vertex for the clause and two vertices for each literal in the clause. If a literal corresponds to an odd-indexed variable or its negation, the literal vertex 
succeeds the clause vertex. If a literal corresponds to an even-indexed variable or its negation, the literal vertex 
succeeds an extra vertex, which succeeds the clause vertex. 
(See Figure~\ref{fig:difclause}.)
\item For each $i>1$, the bottom vertex of the $x_i$ gadget precedes the top vertex of the $x_{i+1}$ gadget.
\item The bottom vertex of the $x_n$ gadget precedes each clause vertex.
\item A vertex in a clause gadget corresponding to a literal $x_i$ for odd $i$ precedes the second vertex in the left path in the $x_i$ gadget. A vertex in a clause gadget corresponding to a literal $\bar{x}_i$ for odd $i$ precedes the second vertex in the right path in the $x_i$ gadget. 
\item A vertex in a clause gadget corresponding to a literal $x_i$ for even $i$ precedes the \emph{right} vertex in the $x_i$ gadget. A vertex in a clause gadget corresponding to a literal $\bar{x}_i$ for even $i$ precedes the \emph{left} vertex in the $x_i$ gadget.
\end{itemize}
See Figure~\ref{fig:dif} for an example of this construction.

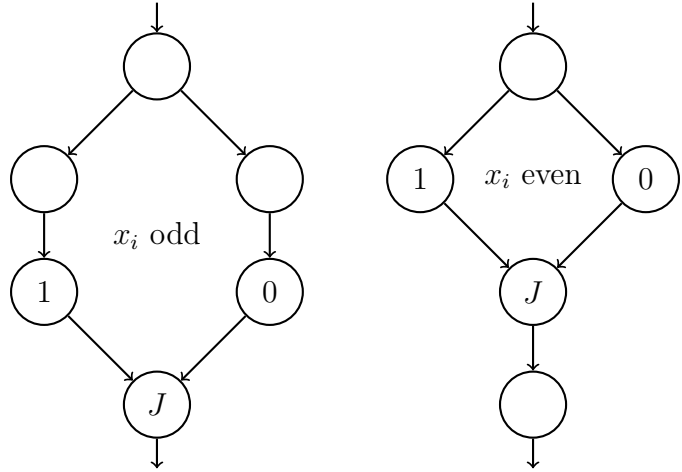
\begin{figure}
\begin{center}
\begin{tikzpicture}[thick, node distance=1.5cm, circ/.style={draw, circle, minimum size=25pt}]

\node[circ] at (0, 0) (top) {};
\node at (0, -2.25cm) (var) {$x_i$ odd};
\node[circ] at (-1.5cm, -1.5cm) (tabove) {};
\node[circ] at (1.5cm, -1.5cm) (fabove) {};
\node[circ] at (-1.5cm, -3cm) (true) {$1$};
\node[circ] at (1.5cm, -3cm) (false) {$0$};
\node[circ] at (0, -4.5cm) (bottom) {$J$};
\node (in) at (0, 1cm) {};
\node (out) at (0, -5.5cm) {};

\draw [->] (top)->(tabove);
\draw [->] (tabove)->(true);
\draw [->] (true)->(bottom);
\draw [->] (top)->(fabove);
\draw [->] (fabove)->(false);
\draw [->] (false)->(bottom);
\draw [->] (in)->(top);
\draw [->] (bottom)->(out);

\node[circ] at (5cm, 0) (tope) {};
\node at (5cm, -1.5cm) (vare) {$x_i$ even};
\node[circ] at (3.5cm, -1.5cm) (truee) {$1$};
\node[circ] at (6.5cm, -1.5cm) (falsee) {$0$};
\node[circ] at (5cm, -3cm) (joine) {$J$};
\node[circ] at (5cm, -4.5cm) (bottome) {};
\node (ine) at (5cm, 1cm) {};
\node (oute) at (5cm, -5.5cm) {};

\draw [->] (tope)->(truee);
\draw [->] (truee)->(joine);
\draw [->] (tope)->(falsee);
\draw [->] (falsee)->(joine);
\draw [->] (joine)->(bottome);
\draw [->] (ine)->(tope);
\draw [->] (bottome)->(oute);
\end{tikzpicture}
\end{center}
\caption{Variable gadgets for \game{DIF}. Symbols $1$ and $0$ denote vertices corresponding to the variable $x_i$ being set to {\ttt} or \fff, respectively. Joining vertices are indicated by the symbol $J$.}
\label{fig:difvar}
\end{figure}

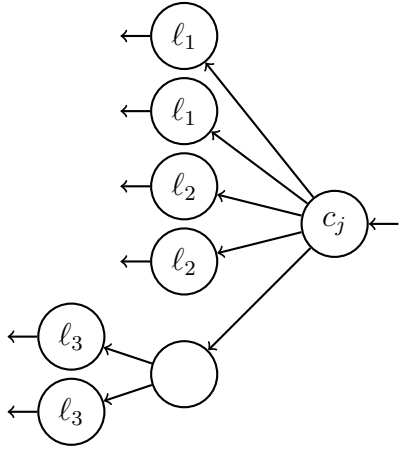
\begin{figure}
\begin{center}
\begin{tikzpicture}[thick, node distance=1.5cm, circ/.style={draw, circle, minimum size=25pt}]

\node at (1cm,0) (in) {};
\node[circ] at (0,0) (clause) {$c_j$};
\node[circ] at (-2cm, 2.5cm) (l1a) {$\ell_1$};
\node[circ] at (-2cm, 1.5cm) (l1b) {$\ell_1$};
\node[circ] at (-2cm, 0.5cm) (l2a) {$\ell_2$};
\node[circ] at (-2cm, -0.5cm) (l2b) {$\ell_2$};
\node[circ] at (-2cm, -2cm) (l3) {};
\node[circ] at (-3.5cm, -1.5cm) (l3ae) {$\ell_3$};
\node[circ] at (-3.5cm, -2.5cm) (l3be) {$\ell_3$};
\node at (-3cm,2.5cm) (out1a) {};
\node at (-3cm,1.5cm) (out1b) {};
\node at (-3cm,0.5cm) (out2a) {};
\node at (-3cm,-0.5cm) (out2b) {};
\node at (-4.5cm,-1.5cm) (out3a) {};
\node at (-4.5cm,-2.5cm) (out3b) {};

\draw [->] (in) -> (clause);
\draw [->] (clause) -> (l1a);
\draw [->] (clause) -> (l1b);
\draw [->] (clause) -> (l2a);
\draw [->] (clause) -> (l2b);
\draw [->] (clause) -> (l3);
\draw [->] (l1a) -> (out1a);
\draw [->] (l1b) -> (out1b);
\draw [->] (l2a) -> (out2a);
\draw [->] (l2b) -> (out2b);
\draw [->] (l3) -> (l3ae);
\draw [->] (l3) -> (l3be);
\draw [->] (l3ae) -> (out3a);
\draw [->] (l3be) -> (out3b);
\end{tikzpicture}
\end{center}
\caption{Clause gadget for \game{DIF}, with clause $c_j=\ell_1\vee\ell_2\vee\ell_3$, if the variable in literal $\ell_3$ has even index and the others have odd index.}
\label{fig:difclause}
\end{figure}

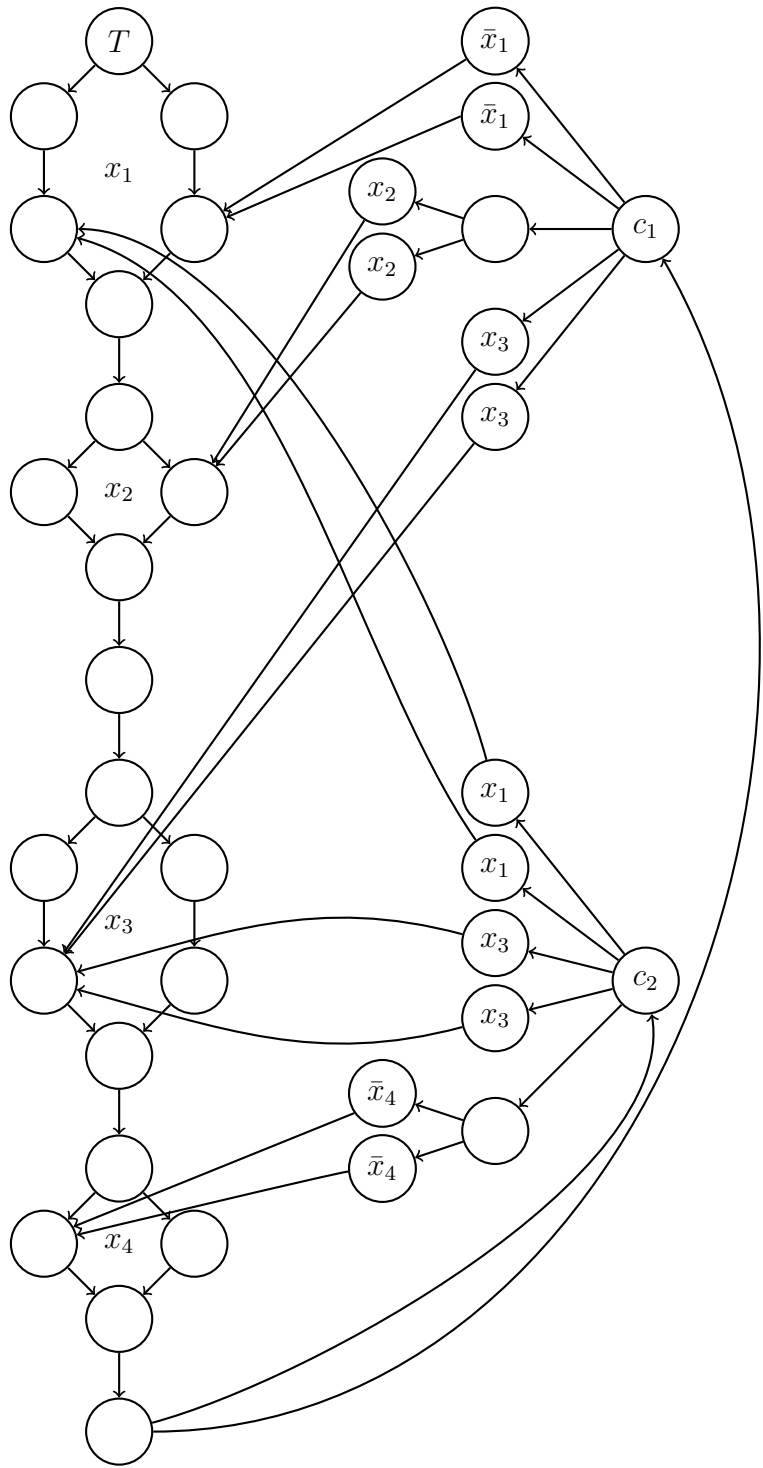
\begin{figure}
\begin{center}
\begin{tikzpicture}[thick, node distance=1.5cm, circ/.style={draw, circle, minimum size=25pt}]

\node[circ] at (0.000000cm, 0.000000cm) (x1top) {$T$};
\node[circ] at (-1.000000cm, -1.000000cm) (x1truep) {};
\node[circ] at (1.000000cm, -1.000000cm) (x1falsep) {};
\node[circ] at (-1.000000cm, -2.500000cm) (x1true) {};
\node[circ] at (1.000000cm, -2.500000cm) (x1false) {};
\node[circ] at (0.000000cm, -3.500000cm) (x1bottom) {};
\node at (0.000000cm, -1.750000cm) (x1var) {$x_{1}$};
\node[circ] at (0.000000cm, -5.000000cm) (x2top) {};
\node[circ] at (-1.000000cm, -6.000000cm) (x2true) {};
\node[circ] at (1.000000cm, -6.000000cm) (x2false) {};
\node[circ] at (0.000000cm, -7.000000cm) (x2join) {};
\node[circ] at (0.000000cm, -8.500000cm) (x2bottom) {};
\node at (0.000000cm, -6.000000cm) (x2var) {$x_{2}$};
\node[circ] at (0.000000cm, -10.000000cm) (x3top) {};
\node[circ] at (-1.000000cm, -11.000000cm) (x3truep) {};
\node[circ] at (1.000000cm, -11.000000cm) (x3falsep) {};
\node[circ] at (-1.000000cm, -12.500000cm) (x3true) {};
\node[circ] at (1.000000cm, -12.500000cm) (x3false) {};
\node[circ] at (0.000000cm, -13.500000cm) (x3bottom) {};
\node at (0.000000cm, -11.750000cm) (x3var) {$x_{3}$};
\node[circ] at (0.000000cm, -15.000000cm) (x4top) {};
\node[circ] at (-1.000000cm, -16.000000cm) (x4true) {};
\node[circ] at (1.000000cm, -16.000000cm) (x4false) {};
\node[circ] at (0.000000cm, -17.000000cm) (x4join) {};
\node[circ] at (0.000000cm, -18.500000cm) (x4bottom) {};
\node at (0.000000cm, -16.000000cm) (x4var) {$x_{4}$};
\node[circ] at (7.000000cm, -2.500000cm) (c1) {$c_{1}$};
\node[circ] at (5.000000cm, 0.000000cm) (c11a) {$\bar{x}_{1}$};
\node[circ] at (5.000000cm, -1.000000cm) (c11b) {$\bar{x}_{1}$};
\node[circ] at (5.000000cm, -2.500000cm) (c12a) {};
\node[circ] at (3.500000cm, -2.000000cm) (c12ae) {$x_{2}$};
\node[circ] at (3.500000cm, -3.000000cm) (c12be) {$x_{2}$};
\node[circ] at (5.000000cm, -4.000000cm) (c13a) {$x_{3}$};
\node[circ] at (5.000000cm, -5.000000cm) (c13b) {$x_{3}$};
\node[circ] at (7.000000cm, -12.500000cm) (c2) {$c_{2}$};
\node[circ] at (5.000000cm, -10.000000cm) (c21a) {$x_{1}$};
\node[circ] at (5.000000cm, -11.000000cm) (c21b) {$x_{1}$};
\node[circ] at (5.000000cm, -12.000000cm) (c22a) {$x_{3}$};
\node[circ] at (5.000000cm, -13.000000cm) (c22b) {$x_{3}$};
\node[circ] at (5.000000cm, -14.500000cm) (c23a) {};
\node[circ] at (3.500000cm, -14.000000cm) (c23ae) {$\bar{x}_{4}$};
\node[circ] at (3.500000cm, -15.000000cm) (c23be) {$\bar{x}_{4}$};

\draw [->] (x1top) -> (x1truep);
\draw [->] (x1top) -> (x1falsep);
\draw [->] (x1truep) -> (x1true);
\draw [->] (x1falsep) -> (x1false);
\draw [->] (x1true) -> (x1bottom);
\draw [->] (x1false) -> (x1bottom);
\draw [->] (x2top) -> (x2true);
\draw [->] (x2top) -> (x2false);
\draw [->] (x2true) -> (x2join);
\draw [->] (x2false) -> (x2join);
\draw [->] (x2join) -> (x2bottom);
\draw [->] (x3top) -> (x3truep);
\draw [->] (x3top) -> (x3falsep);
\draw [->] (x3truep) -> (x3true);
\draw [->] (x3falsep) -> (x3false);
\draw [->] (x3true) -> (x3bottom);
\draw [->] (x3false) -> (x3bottom);
\draw [->] (x4top) -> (x4true);
\draw [->] (x4top) -> (x4false);
\draw [->] (x4true) -> (x4join);
\draw [->] (x4false) -> (x4join);
\draw [->] (x4join) -> (x4bottom);
\draw [->] (x1bottom) -> (x2top);
\draw [->] (x2bottom) -> (x3top);
\draw [->] (x3bottom) -> (x4top);
\draw [->] (x4bottom) to [out=0, in=300] (c1);
\draw [->] (c1) -> (c11a);
\draw [->] (c1) -> (c11b);
\draw [->] (c11a) -> (x1false);
\draw [->] (c11b) to [out=180, in=20, looseness=0] (x1false);
\draw [->] (c1) -> (c12a);
\draw [->] (c12a) -> (c12ae);
\draw [->] (c12a) -> (c12be);
\draw [->] (c12ae) -> (x2false);
\draw [->] (c12be) -> (x2false);
\draw [->] (c1) -> (c13a);
\draw [->] (c1) -> (c13b);
\draw [->] (c13a) -> (x3true);
\draw [->] (c13b) -> (x3true);
\draw [->] (x4bottom) to [out=15, in=280, looseness=0.65] (c2);
\draw [->] (c2) -> (c21a);
\draw [->] (c2) -> (c21b);
\draw [->] (c21a) to [out=105, in=0, looseness=0.65] (x1true);
\draw [->] (c21b) to [out=125, in=345, looseness=0.75] (x1true);
\draw [->] (c2) -> (c22a);
\draw [->] (c2) -> (c22b);
\draw [->] (c22a) to [out=165, in=15, looseness=1] (x3true);
\draw [->] (c22b) to [out=195, in=-15, looseness=1] (x3true);
\draw [->] (c2) -> (c23a);
\draw [->] (c23a) -> (c23ae);
\draw [->] (c23a) -> (c23be);
\draw [->] (c23ae) to [out=215, in=30, looseness=0] (x4true);
\draw [->] (c23be) to [out=185, in=15, looseness=0.1] (x4true);
\end{tikzpicture}
\end{center}
\caption{\game{DIF} reduction from \game{TQBF} instance $\exists x_1\forall x_2\exists x_3\forall x_4\pb{\bar{x}_1\vee x_2\vee x_3}\wedge\pb{x_1\vee x_3\vee\bar{x}_4}$}
\label{fig:dif}
\end{figure}

Let us analyze how the game progresses. First, the token moves through the variable gadgets, then the token moves through a single clause gadget. After the token reaches a literal vertex in a clause gadget, there is at most one more move in the game.

First, let us examine the play options through the variable gadgets. When the token is at the top vertex of a variable gadget, it is Left's turn. 
Several move/deletion options exist as the token traverses the variable gadget, but we shall observe that only one decision affects the outcome: which vertex is deleted when the token is moved to the joining vertex. This decision is key, and it is important to note that it is made by Left for odd-indexed variables and by Right for even-indexed variables. After the bottom of the variable gadget is reached, if more variables remain Right has a forced regular move from the bottom of the current gadget to the top of the next gadget. Observe that along the way, every joining vertex is deleted.

After all of the variable gadgets have been traversed, the token is at the bottom vertex of the $x_n$ gadget with Right to move. He can move regularly to any clause vertex. Left then can select a literal in that clause and move regularly to a vertex for it if it has odd index or to the extra vertex for it if it has even index. Note that at least one instance of each literal vertex is available, as any pair of duplicate literal vertices share a single successor. 
It is now Right's turn. Suppose first that the chosen literal contains an odd-indexed variable. If the successor of the current vertex has been deleted, the game is over and Left wins. If not, Right moves to that vertex and deletes any of its predecessors. Left then loses, as this vertex is the last vertex in a left or right path in an odd-indexed variable gadget, and all joining vertices (including this vertex's only successor) have already been deleted. On the other hand, suppose that the chosen literal contains  an even-indexed variable. Right is now forced to move regularly to a vertex for that literal, and it becomes Left's turn. If the successor of the current vertex has been deleted, the game is over and Right wins. If not, Left moves to that vertex and deletes any of its predecessors. Right then loses, as this vertex is the last vertex in a left or right path in an even-indexed variable gadget, and all joining vertices (including this vertex's only successor) have already been deleted.

We now claim that $Q$ is true if and only if Left wins the corresponding game of \game{DIF}. First, suppose $Q$ is true. Then, there is some setting of the variables $x_1,x_3,\ldots,x_{n-1}$ so that the formula is true regardless of the values of $x_2,x_4,\ldots,x_n$. Take such a setting, and, as Left, adopt the following strategy:
\begin{itemize}
\item When moving to the joining vertex of the gadget for $x_i$, delete the left predecessor if $x_i$ should be set to \ttt, and delete the right predecessor if $x_i$ should be set to \fff.
\item When moving from a clause vertex, choose a successor corresponding to a true literal.
\end{itemize}
Since $Q$ is true, every clause is true, which means that each clause contains a true literal. Right then loses, as the successor of the chosen literal vertex has been deleted earlier in the game if it contains an odd-indexed variable, or the successor of the successor of the literal vertex that Right moves to has not been deleted earlier in the game if it contains an even-indexed variable.

Now, suppose $Q$ is false. Then, there is some setting of the variables $x_2,x_4,\ldots,x_n$ so that the formula is false regardless of the values of $x_1,x_3,\ldots,x_{n-1}$. Take such a setting, and, as Right, adopt the following strategy:
\begin{itemize}
\item When moving to the joining vertex of the gadget for $x_i$, delete the left predecessor if $x_i$ should be set to \ttt, and delete the right predecessor if $x_i$ should be set to \fff.
\item When moving to a clause vertex, choose a clause vertex corresponding to a false clause.
\end{itemize}
Since $Q$ is false, there is a false clause, which, as a disjunction of literals, contains no true literal. Right then wins, as the successor of the literal vertex that Left moves to is still present in the graph if it contains an odd-indexed variable, or the successor of the literal vertex that Right moves to is no longer present in the graph if it contains an even-indexed variable.

This completes the reduction showing that $\text{\game{TQBF}}\leq_p DIF$. Therefore, \game{DIF} is \ps-hard, and, hence, \ps-complete, as required.

What remains is to prove that the resulting graph is bipartite (ignoring direction). We observe that the following describes a bipartition $\pb{A,B}$ of the graph.
\begin{description}
\item[$A$:]\lb
\begin{itemize}
\item Top vertices of variable gadgets.
\item Second vertices in left and right paths of variable gadgets corresponding to odd-indexed variables.
\item Joining vertices in variable gadgets corresponding to even-indexed variables.
\item Clause vertices.
\item Literal vertices for even-indexed variables.
\end{itemize}
\item[$B$:]\lb
\begin{itemize}
\item First vertices in left and right paths of variable gadgets corresponding to odd-indexed variables.
\item Left and right vertices of variable gadgets corresponding to even-indexed variables.
\item Bottom vertices of variable gadgets.
\item Literal vertices for odd-indexed variables.
\item Extra vertices in clause gadgets.
\end{itemize}
\end{description}
\end{proof}

\subsection{Complexity of DPF}\label{ss:dpf}

In Section~\ref{ss:upf}, we prove that the game \game{UPF} is \ps-complete. By Proposition~\ref{prop:udtriv}, this implies that \game{DPF} is \ps-complete. But, we include a separate proof here, as this reduction is simpler and yields an underlying bipartite graph.
\begin{theorem}\label{thm:dpf}
The game \game{DPF} is \ps-complete, even when restricted to bipartite graphs.
\end{theorem}
\begin{proof}
We proceed via a reduction from \game{TQBF} to \game{DPF}. Let $Q=\exists x_1\forall x_2\cdots\exists x_{n-1}\forall x_n c_1\wedge c_2\wedge\cdots\wedge c_m$ be an instance of \game{TQBF}. Without loss of generality, assume that $m\geq4$ and that $m$ is even (add copies of clauses if this is not the case). We construct an instance of \game{DPF} as follows:
\begin{itemize}
\item For each variable $x_i$, there is a diamond structure called a \emph{variable gadget}, with the top vertex preceding each of the left and right vertices, each of which precedes the bottom vertex. 
(See Figure~\ref{fig:dpfvar}.)
\item For each clause $c_j$, there is a single vertex.
\item Left's token starts at the top vertex of the $x_1$ gadget. Right's token starts at the top vertex of the $x_2$ gadget.
\item For $1\leq i<\frac{n}{2}$, the bottom vertex of the $x_{2i-1}$ gadget is also the top vertex of the $x_{2i+1}$ gadget, and the bottom vertex of the $x_{2i}$ gadget is also the top vertex of the $x_{2i+2}$ gadget.
\item There is a path of $2m-3$ vertices, called \emph{delay vertices}. The first delay vertex succeeds the bottom vertex of the $x_{n-1}$ gadget, and the last delay vertex precedes each clause vertex.
\item There is a path of $2m-3$ vertices, called \emph{clause deletion vertices}. The first clause deletion vertex succeeds the bottom vertex of the $x_n$ gadget. Each odd-indexed clause deletion vertex succeeds each clause vertex.
\item There is a path of $m-1$ vertices, called \emph{escape vertices}. The first escape vertex succeeds the last clause deletion vertex.
\item For each literal $x_i$ appearing in clause $c_j$, there is a  path of $m-3$ vertices succeeding $c_j$. The last vertex in this path precedes two vertices, each of which precede the right vertex in the gadget for $x_i$. For each literal $\bar{x}_i$ appearing in clause $c_j$, there is a path of $m-3$ vertices succeeding $c_j$. The last vertex in this path precedes two vertices, each of which precede the left vertex in the gadget for $x_i$. We refer to each such structure linking a clause vertex back to a variable gadget as a \emph{linker}.
\end{itemize}
See Figure~\ref{fig:dpf} for a schematic of this construction. Delay vertices are marked by \textsc{dly}. Clause deletion vertices are marked by \textsc{dlt}. Escape vertices are marked by \textsc{esc}. Linker vertices are marked by \textsc{lnk}. There are several linkers that are not pictured; only the linkers from $c_1$ are shown.

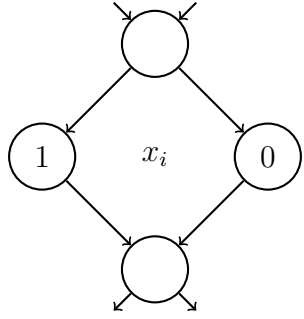
\begin{figure}
\begin{center}
\begin{tikzpicture}[thick, node distance=1.5cm, circ/.style={draw, circle, minimum size=25pt}]

\node[circ] at (0, 0) (top) {};
\node at (0, -1.5cm) (var) {$x_i$};
\node[circ] at (-1.5cm, -1.5cm) (true) {$1$};
\node[circ] at (1.5cm, -1.5cm) (false) {$0$};
\node[circ] at (0, -3cm) (bottom) {};
\node (in1) at (-0.7cm, 0.7cm) {};
\node (in2) at (0.7cm, 0.7cm) {};
\node (out1) at (-0.7cm, -3.7cm) {};
\node (out2) at (0.7cm, -3.7cm) {};

\draw [->] (top)->(true);
\draw [->] (true)->(bottom);
\draw [->] (top)->(false);
\draw [->] (false)->(bottom);
\draw [->] (in1)->(top);
\draw [->] (in2)->(top);
\draw [->] (bottom)->(out1);
\draw [->] (bottom)->(out2);
\end{tikzpicture}
\end{center}
\caption{Variable gadget for \game{DPF}. Symbols $1$ and $0$ denote vertices corresponding to the variable $x_i$ being set to {\ttt} or \fff, respectively.}
\label{fig:dpfvar}
\end{figure}

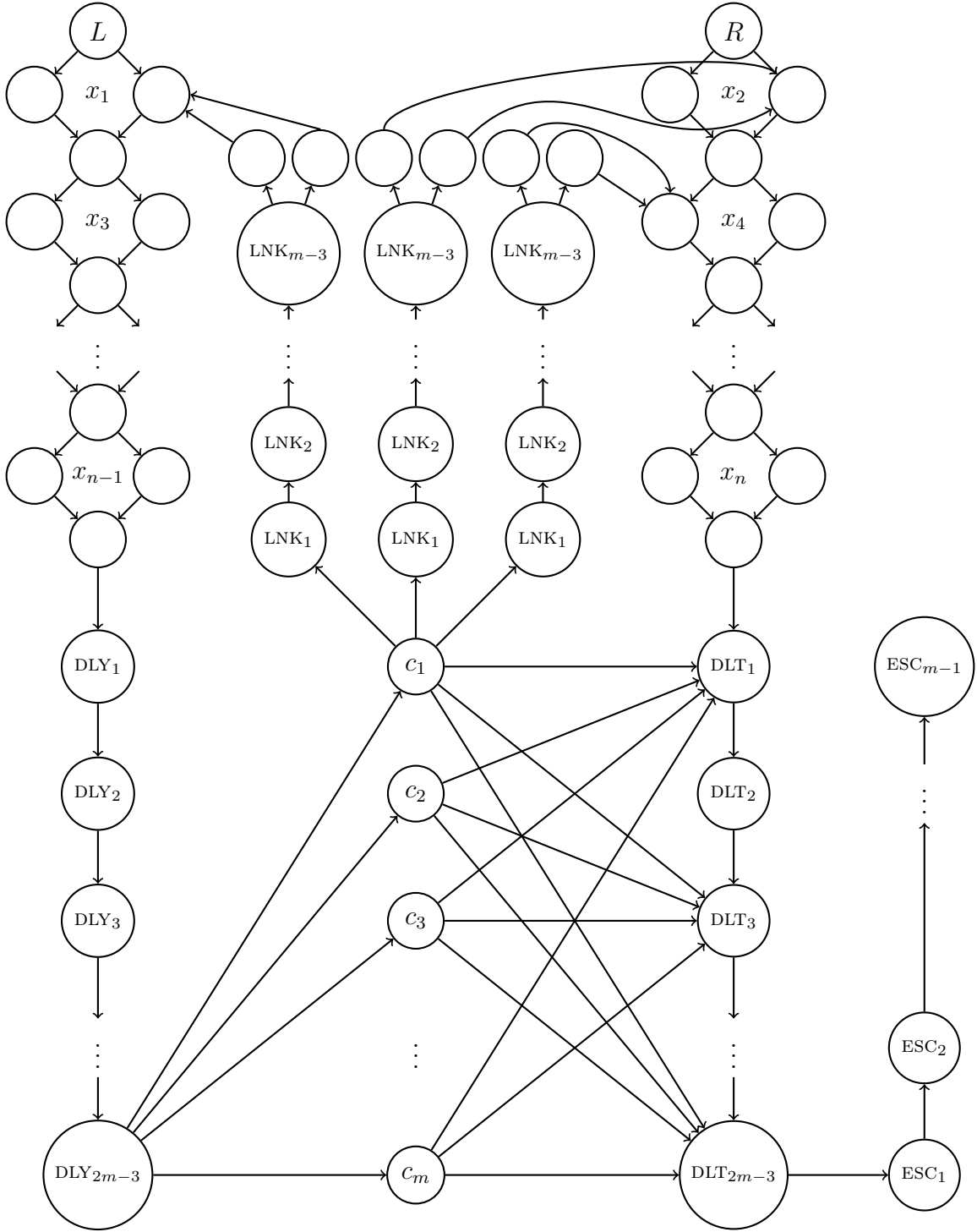
\begin{figure}
\begin{center}
\begin{tikzpicture}[thick, node distance=1.5cm, circ/.style={draw, circle, minimum size=25pt}]

\node[circ, align=center] at (0.000000cm, 0.000000cm) (x1top) {$L$};
\node[circ] at (-1.000000cm, -1.000000cm) (x1true) {};
\node[circ] at (1.000000cm, -1.000000cm) (x1false) {};
\node at (0.000000cm, -1.000000cm) (x1var) {$x_{1}$};
\node[circ, align=center] at (10.000000cm, 0.000000cm) (x2top) {$R$};
\node[circ] at (9.000000cm, -1.000000cm) (x2true) {};
\node[circ] at (11.000000cm, -1.000000cm) (x2false) {};
\node at (10.000000cm, -1.000000cm) (x2var) {$x_{2}$};
\node[circ] at (0.000000cm, -2.000000cm) (x3top) {};
\node[circ] at (-1.000000cm, -3.000000cm) (x3true) {};
\node[circ] at (1.000000cm, -3.000000cm) (x3false) {};
\node at (0.000000cm, -3.000000cm) (x3var) {$x_{3}$};
\node[circ] at (10.000000cm, -2.000000cm) (x4top) {};
\node[circ] at (9.000000cm, -3.000000cm) (x4true) {};
\node[circ] at (11.000000cm, -3.000000cm) (x4false) {};
\node at (10.000000cm, -3.000000cm) (x4var) {$x_{4}$};
\node[circ] at (0.000000cm, -4.000000cm) (lfakebottom) {};
\node[circ] at (10.000000cm, -4.000000cm) (rfakebottom) {};
\node at (-0.800000cm, -4.800000cm) (llout) {};
\node at (0.800000cm, -4.800000cm) (lrout) {};
\node at (0.000000cm, -5.000000cm) (ldots) {$\vdots$};
\node at (9.200000cm, -4.800000cm) (rlout) {};
\node at (10.800000cm, -4.800000cm) (rrout) {};
\node at (10.000000cm, -5.000000cm) (rdots) {$\vdots$};
\node[circ] at (0.000000cm, -6.000000cm) (xn1top) {};
\node[circ] at (-1.000000cm, -7.000000cm) (xn1true) {};
\node[circ] at (1.000000cm, -7.000000cm) (xn1false) {};
\node at (0.000000cm, -7.000000cm) (xn1var) {$x_{n-1}$};
\node[circ] at (10.000000cm, -6.000000cm) (xntop) {};
\node[circ] at (9.000000cm, -7.000000cm) (xntrue) {};
\node[circ] at (11.000000cm, -7.000000cm) (xnfalse) {};
\node at (10.000000cm, -7.000000cm) (xnvar) {$x_{n}$};
\node at (-0.800000cm, -5.200000cm) (llin) {};
\node at (0.800000cm, -5.200000cm) (lrin) {};
\node at (9.200000cm, -5.200000cm) (rlin) {};
\node at (10.800000cm, -5.200000cm) (rrin) {};
\node[circ] at (0.000000cm, -8.000000cm) (lbottom) {};
\node[circ] at (10.000000cm, -8.000000cm) (rbottom) {};
\node[circ] at (0.000000cm, -10.000000cm) (dly1) {{\footnotesize $\text{\textsc{dly}}_{1}$}};
\node[circ] at (10.000000cm, -10.000000cm) (cd1) {{\footnotesize $\text{\textsc{dlt}}_{1}$}};
\node[circ] at (0.000000cm, -12.000000cm) (dly2) {{\footnotesize $\text{\textsc{dly}}_{2}$}};
\node[circ] at (10.000000cm, -12.000000cm) (cd2) {{\footnotesize $\text{\textsc{dlt}}_{2}$}};
\node[circ] at (0.000000cm, -14.000000cm) (dly3) {{\footnotesize $\text{\textsc{dly}}_{3}$}};
\node[circ] at (10.000000cm, -14.000000cm) (cd3) {{\footnotesize $\text{\textsc{dlt}}_{3}$}};
\node at (0.000000cm, -16.000000cm) (dlydots) {$\vdots$};
\node[circ] at (0.000000cm, -18.000000cm) (dlym1) {{\footnotesize $\text{\textsc{dly}}_{2m-3}$}};
\node at (10.000000cm, -16.000000cm) (cddots) {$\vdots$};
\node[circ] at (10.000000cm, -18.000000cm) (cdm1) {{\footnotesize $\text{\textsc{dlt}}_{2m-3}$}};
\node[circ] at (5.000000cm, -10.000000cm) (c1) {$c_{1}$};
\node[circ] at (5.000000cm, -12.000000cm) (c2) {$c_{2}$};
\node[circ] at (5.000000cm, -14.000000cm) (c3) {$c_{3}$};
\node at (5.000000cm, -16.000000cm) (clausedots) {$\vdots$};
\node[circ] at (5.000000cm, -18.000000cm) (cm) {$c_m$};
\node[circ] at (13.000000cm, -18.000000cm) (e1) {{\footnotesize $\text{\textsc{esc}}_{1}$}};
\node[circ] at (13.000000cm, -16.000000cm) (e2) {{\footnotesize $\text{\textsc{esc}}_{2}$}};
\node at (13.000000cm, -12.000000cm) (edots) {$\vdots$};
\node[circ] at (13.000000cm, -10.000000cm) (em1) {{\footnotesize $\text{\textsc{esc}}_{m-1}$}};
\node[circ] at (3.000000cm, -8.000000cm) (l1i1) {{\footnotesize $\text{\textsc{lnk}}_{1}$}};
\node[circ] at (3.000000cm, -6.500000cm) (l1i2) {{\footnotesize $\text{\textsc{lnk}}_{2}$}};
\node at (3.000000cm, -5.000000cm) (l1dots) {$\vdots$};
\node[circ] at (3.000000cm, -3.500000cm) (l1m3) {{\footnotesize $\text{\textsc{lnk}}_{m-3}$}};
\node[circ] at (2.500000cm, -2.000000cm) (l1m2a) {};
\node[circ] at (3.500000cm, -2.000000cm) (l1m2b) {};
\node[circ] at (5.000000cm, -8.000000cm) (l2i1) {{\footnotesize $\text{\textsc{lnk}}_{1}$}};
\node[circ] at (5.000000cm, -6.500000cm) (l2i2) {{\footnotesize $\text{\textsc{lnk}}_{2}$}};
\node at (5.000000cm, -5.000000cm) (l2dots) {$\vdots$};
\node[circ] at (5.000000cm, -3.500000cm) (l2m3) {{\footnotesize $\text{\textsc{lnk}}_{m-3}$}};
\node[circ] at (4.500000cm, -2.000000cm) (l2m2a) {};
\node[circ] at (5.500000cm, -2.000000cm) (l2m2b) {};
\node[circ] at (7.000000cm, -8.000000cm) (l3i1) {{\footnotesize $\text{\textsc{lnk}}_{1}$}};
\node[circ] at (7.000000cm, -6.500000cm) (l3i2) {{\footnotesize $\text{\textsc{lnk}}_{2}$}};
\node at (7.000000cm, -5.000000cm) (l3dots) {$\vdots$};
\node[circ] at (7.000000cm, -3.500000cm) (l3m3) {{\footnotesize $\text{\textsc{lnk}}_{m-3}$}};
\node[circ] at (6.500000cm, -2.000000cm) (l3m2a) {};
\node[circ] at (7.500000cm, -2.000000cm) (l3m2b) {};

\draw [->] (x1top) -> (x1true);
\draw [->] (x1top) -> (x1false);
\draw [->] (x2top) -> (x2true);
\draw [->] (x2top) -> (x2false);
\draw [->] (x1true) -> (x3top);
\draw [->] (x1false) -> (x3top);
\draw [->] (x3top) -> (x3true);
\draw [->] (x3top) -> (x3false);
\draw [->] (x2true) -> (x4top);
\draw [->] (x2false) -> (x4top);
\draw [->] (x4top) -> (x4true);
\draw [->] (x4top) -> (x4false);
\draw [->] (x3true) -> (lfakebottom);
\draw [->] (x3false) -> (lfakebottom);
\draw [->] (x4true) -> (rfakebottom);
\draw [->] (x4false) -> (rfakebottom);
\draw [->] (lfakebottom) -> (llout);
\draw [->] (lfakebottom) -> (lrout);
\draw [->] (rfakebottom) -> (rlout);
\draw [->] (rfakebottom) -> (rrout);
\draw [->] (xn1top) -> (xn1true);
\draw [->] (xn1top) -> (xn1false);
\draw [->] (xntop) -> (xntrue);
\draw [->] (xntop) -> (xnfalse);
\draw [->] (llin) -> (xn1top);
\draw [->] (lrin) -> (xn1top);
\draw [->] (rlin) -> (xntop);
\draw [->] (rrin) -> (xntop);
\draw [->] (xn1true) -> (lbottom);
\draw [->] (xn1false) -> (lbottom);
\draw [->] (xntrue) -> (rbottom);
\draw [->] (xnfalse) -> (rbottom);
\draw [->] (dly2) -> (dly3);
\draw [->] (cd1) -> (cd2);
\draw [->] (cd2) -> (cd3);
\draw [->] (lbottom) -> (dly1);
\draw [->] (dly1) -> (dly2);
\draw [->] (dly3) -> (dlydots);
\draw [->] (dlydots) -> (dlym1);
\draw [->] (rbottom) -> (cd1);
\draw [->] (cd3) -> (cddots);
\draw [->] (cddots) -> (cdm1);
\draw [->] (c1) -> (cd1);
\draw [->] (c1) -> (cd3);
\draw [->] (c1) -> (cdm1);
\draw [->] (dlym1) -> (c1);
\draw [->] (c2) -> (cd1);
\draw [->] (c2) -> (cd3);
\draw [->] (c2) -> (cdm1);
\draw [->] (dlym1) -> (c2);
\draw [->] (c3) -> (cd1);
\draw [->] (c3) -> (cd3);
\draw [->] (c3) -> (cdm1);
\draw [->] (dlym1) -> (c3);
\draw [->] (cm) -> (cd1);
\draw [->] (cm) -> (cd3);
\draw [->] (cm) -> (cdm1);
\draw [->] (dlym1) -> (cm);
\draw [->] (e1) -> (e2);
\draw [->] (cdm1) -> (e1);
\draw [->] (e2) -> (edots);
\draw [->] (edots) -> (em1);
\draw [->] (l1i1) -> (l1i2);
\draw [->] (c1) -> (l1i1);
\draw [->] (l1i2) -> (l1dots);
\draw [->] (l1dots) -> (l1m3);
\draw [->] (l1m3) -> (l1m2a);
\draw [->] (l1m3) -> (l1m2b);
\draw [->] (l2i1) -> (l2i2);
\draw [->] (c1) -> (l2i1);
\draw [->] (l2i2) -> (l2dots);
\draw [->] (l2dots) -> (l2m3);
\draw [->] (l2m3) -> (l2m2a);
\draw [->] (l2m3) -> (l2m2b);
\draw [->] (l3i1) -> (l3i2);
\draw [->] (c1) -> (l3i1);
\draw [->] (l3i2) -> (l3dots);
\draw [->] (l3dots) -> (l3m3);
\draw [->] (l3m3) -> (l3m2a);
\draw [->] (l3m3) -> (l3m2b);

\draw [->] (l1m2a) to (x1false);
\draw [->] (l1m2b) to [out=90, in=0, looseness=0] (x1false);
\draw [->] (l2m2a) to [out=90, in=135, looseness=0.3] (x2false);
\draw [->] (l2m2b) to [out=45, in=210, looseness=1.15] (x2false);
\draw [->] (l3m2a) to [out=55, in=90, looseness=0.7] (x4true);
\draw [->] (l3m2b) to (x4true);
\end{tikzpicture}
\end{center}
\caption{Schematic diagram \game{DPF} construction from \game{TQBF} instance $\exists x_1\forall x_2\cdots\exists x_{n-1}\forall x_n\pb{x_1\vee x_2\vee\bar{x}_4}\wedge c_2\wedge\cdots\wedge c_m$. Linkers succeeding $c_2$ through $c_m$ exist but are not shown.}
\label{fig:dpf}
\end{figure}

Let us analyze how the game progresses. First, the tokens moves through the variable gadgets. Then, Left's token traverses the delay vertices while Right's token traverses the clause deletion vertices. After $m-1$ moves by each player, both tokens are at the bottom of their respective structures, with Left to move. Left moves to a clause vertex. We observe later how  the game proceeds from here.

First, let us examine the play options through the variable gadgets. When the active player's token is at the top vertex of the gadget for $x_i$, the active player has two move options (the left and right vertices in the gadget), and that player can delete the top vertex of the gadget or any vertex in a linker preceding the vertex moved to. Neither this move nor this deletion affects the outcome. The move after this decision is forced, to the bottom of the gadget, but either the left or right vertex of the gadget can be deleted. This decision is key, and it is important to note that these decisions are made in order of the index of the variables by alternating players, starting with Left.

After all of the variable gadgets have been traversed, Left's token is at the bottom of the $x_{n-1}$ gadget, and Right's token is at the bottom of the $x_n$ gadget, with Left to move. Her next $2m-3$ moves are forced to regularly traverse the chain of delay vertices. 
Meanwhile, Right makes $2m-3$ moves through the clause deletion vertices. After each move to an odd-indexed clause deletion vertex, he can delete any remaining clause vertex, or he can delete the vertex he moved from. Not deleting a clause when able only serves to give Left more options later in the game, so Right can do no better than to always delete clause vertices. Right's moves to even-indexed clause deletion vertices are forced to be regular. When Left is to move from the last delay vertex, Right's token is at the last clause deletion vertex, and there is one clause vertex that has not been deleted. Left is forced to move regularly to that vertex. Then, Right is forced to move regularly to the first escape vertex (deleting the last clause deletion vertex). 
Now, all of the odd-indexed clause deletion vertices have been deleted, so Left must move to a linker.

Once in a linker, Left's next $m-4$ moves are regular and forced along a path. The move after that is also regular to one of the two last vertices of the linker, at least one of which remains. At that point, Left is out of moves if the vertex in the variable gadget succeeding the linker has been deleted and can move to that vertex otherwise. So, Left has $m-3$ moves remaining if she chooses a linker that no longer connects to a variable gadget, and she has at least $m-2$ moves remaining if the linker does still connect. Since Right moves next with $m-2$ moves remaining in the escape path, Left wins if and only if the linker still connects to a variable gadget.

We now claim that $Q$ is true if and only if Left wins the corresponding game of \game{DPF}. First, suppose $Q$ is true. Then, there is some setting of the variables $x_1,x_3,\ldots,x_{n-1}$ so that the formula is true regardless of the values of $x_2,x_4,\ldots,x_n$. Take such a setting, and, as Left, adopt the following strategy:
\begin{itemize}
\item When moving to the bottom vertex of the gadget for $x_i$, delete the left vertex if $x_i$ should be set to \ttt, and delete the right vertex if $x_i$ should be set to \fff.
\item When moving from a clause vertex, enter a linker corresponding to a true literal.
\end{itemize}
Since $Q$ is true, at least one of the remaining clause's literals is true. The successor of the linker from that clause vertex corresponding to that literal has not been deleted.

Now, suppose $Q$ is false. Then, there is some setting of the variables $x_2,x_4,\ldots,x_n$ so that the formula is false regardless of the values of $x_1,x_3,\ldots,x_{n-1}$. Take such a setting, and, as Right, adopt the following strategy:
\begin{itemize}
\item When moving to the bottom vertex of the gadget for $x_i$, delete the left vertex if $x_i$ should be set to \ttt, and delete the right vertex if $x_i$ should be set to \fff.
\item For some $j$ such that clause $c_j$ is false, avoid deleting the vertex corresponding to $c_j$ (so that Left is forced to move there).
\end{itemize}
As we have seen, later in the game, Left moves to the clause vertex corresponding to $c_j$. This clause is false, so all of its literals are false. So, the successors of every linker from that clause vertex have been deleted, so Right wins.

This completes the reduction showing that $\text{\game{TQBF}}\leq_p DPF$. Therefore, \game{DPF} is \ps-hard, and, hence, \ps-complete, as required.

Finally, we claim that the underlying graph in the \game{DPF} instance is bipartite. Consider the partition $\pb{A,B}$ as follows:
\begin{description}
\item[$A$:]\lb
\begin{itemize}
\item Top and bottom vertices of variable gadgets.
\item Even-indexed delay vertices.
\item Even-indexed clause deletion vertices.
\item Clause vertices.
\item Both vertices in each linker that precede a node in a variable gadget.
\item Remaining linker vertices that are not in $B$.
\item Odd-indexed escape vertices.
\end{itemize}
\item[$B$:]\lb
\begin{itemize}
\item Left and right vertices of variable gadgets.
\item Odd-indexed delay vertices.
\item Odd-indexed clause deletion vertices.
\item Alternating vertices along the path in each linker, starting with the vertex that succeeds a clause vertex.
\item Even-indexed escape vertices.
\end{itemize}
\end{description}
We claim that $\pb{A,B}$ is a bipartition. We check a few potential problem spots; the rest of the claim is obvious. First, since $2m-3$ is odd, the last delay vertex is in $B$. This is consistent with the clause vertices being in $A$. Similarly, the clause vertices only precede odd-indexed clause deletion vertices, all of which are in $B$. Finally, $m-3$ is odd (since $m$ is even), so the linker vertices that precede two vertices in a linker are in $B$. Their successors are in $A$, and the successors of those vertices, as left or right vertices in variable gadgets, are in $B$, as required.
\end{proof}

\section{Undirected Partizan Geography}\label{s:up}

In this section, we analyze the games \game{UPR} and \game{UPF}, both of which are \ps-complete.

\subsection{Complexity of UPR}\label{ss:upr}

The game \game{UPR} is often referred to as \game{Undirected Partizan Geography}. It is previously known to be \textsc{np}-hard, via a reduction from the Hamiltonian path problem~\cite{fsugeog}. We prove now that \game{Undirected Partizan Geography} is, in fact, \ps-hard, even when restricted to bipartite graphs.

\begin{theorem}\label{thm:upr}
The game \game{UPR} is \ps-complete, even when restricted to bipartite graphs.
\end{theorem}
\begin{proof}
We proceed via a reduction from \game{TQBF} to \game{UPR} on bipartite graphs. Let $Q=\exists x_1\forall x_2\cdots\exists x_{n-1}\forall x_n c_1\wedge c_2\wedge\cdots\wedge c_m$ be an instance of \game{TQBF}. Without loss of generality, assume that $n\geq4$, that $m\geq2$, and that every literal appears in some clause. (If this is not the case, append a cause of the form $x_i\vee x_i\vee\bar{x}_i$ for any literal $x_i$ or $\bar{x}_i$ not appearing.) We construct an instance of \game{UPR} as follows.
\begin{itemize}
\item For each variable $x_i$, there is a diamond structure called a \emph{variable gadget}, with the top vertex adjacent to each of the left and right vertices, each of which is adjacent to the bottom vertex (much like the variable gadget for \game{DPF}, but with undirected edges).
\item Left's token starts at the top vertex of the $x_1$ gadget. Right's token starts at the top vertex of the $x_2$ gadget.
\item For $1\leq i<\frac{n}{2}$, the bottom vertex of the $x_{2i-1}$ gadget is also the top vertex of the $x_{2i+1}$ gadget, and the bottom vertex of the $x_{2i}$ gadget is also the top vertex of the $x_{2i+2}$ gadget.
\item There is a complete bipartite graph $K_{m,m-1}$ called the \emph{clause selection gadget}. The vertices on the side with $m$ vertices correspond to the clauses; these are called \emph{clause vertices}. (See Figure~\ref{fig:uprcl}.) 
\item Each clause vertex is adjacent to the bottom vertex of the $x_n$ gadget.
\item There is a complete bipartite graph $K_{m+\frac{n}{2}+5,m+\frac{n}{2}+5}$ called the \emph{delay graph}, with its two parts known as \emph{Part~I} and \emph{Part~II}. All vertices in Part~I are adjacent to the bottom vertex of the $x_{n-1}$ gadget.
\item There is a vertex called \textsc{exit} that is adjacent to the bottom vertex of the $x_{n-1}$ gadget and adjacent to every vertex in Part~II of the delay graph.
\item There is a path linking \textsc{exit} to each clause vertex. There are $n+4$ vertices along this path not including \textsc{exit} or the clause vertex. Each such path is known as a \emph{clause connector}.
\item For each literal $x_i$ appearing in clause $c_j$, there are two independent paths from the vertex for $c_j$ to the right vertex in the gadget for $x_i$. For each literal $\bar{x}_i$ appearing in clause $c_j$, there are two independent paths from the vertex for $c_j$ to the left vertex in the gadget for $x_i$. If $i$ is odd, there are $n+2$ vertices along each path, not including the clause vertex or the vertex in the variable gadget. If $i$ is even, there are $n+7$ vertices along each path, not including the clause vertex or the vertex in the variable gadget.  We refer to each such path linking a clause vertex back to a variable gadget as a \emph{linker}. (Each clause vertex has six linkers attached to it). If $i$ is odd, we call the path a \emph{left-linker}, and if $i$ is even, we call it a \emph{right-linker}.
\item There is a path of $3n+25$ vertices, known as \emph{escape vertices}. The first escape vertex is adjacent to each non-clause vertex in the clause selection gadget. 
\end{itemize}
See Figure~\ref{fig:upr} for a schematic of this construction. Escape vertices are marked by \textsc{esc}. The clause selection gadget is the complete bipartite graph below Right's variable gadgets. The delay graph is the complete bipartite graph below Left's variable gadgets. The clause connectors are the paths between \textsc{exit} and the clauses. The paths connecting $c_1$ to the right vertex of the $x_1$ gadget are left-linkers and contain $n+2$ vertices each. The paths connecting $c_1$ to other variable gadgets are right-linkers and contain $n+7$ vertices each. There are several linkers that are not pictured; only the linkers adjacent to $c_1$ are shown.

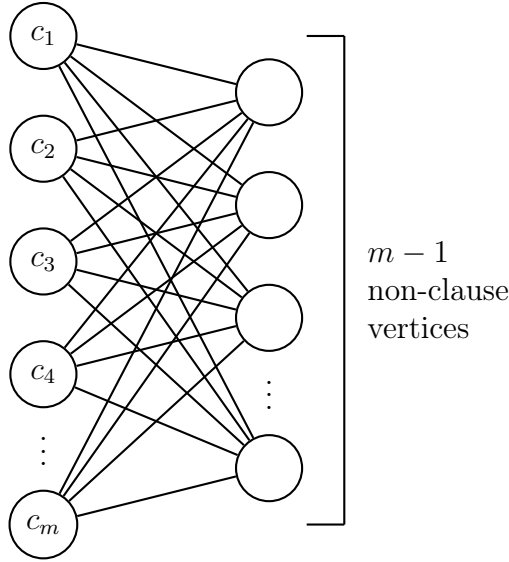
\begin{figure}
\begin{center}
\begin{tikzpicture}[thick, node distance=1.5cm, circ/.style={draw, circle, minimum size=25pt}]

\node[circ] at (0,0) (c1) {$c_1$};
\node[circ] at (0,-1.5) (c2) {$c_2$};
\node[circ] at (0,-3) (c3) {$c_3$};
\node[circ] at (0,-4.5) (c4) {$c_4$};
\node at (0,-5.4) (cdots) {$\vdots$};
\node[circ] at (0,-6.5) (cm) {$c_m$};

\node[circ] at (3,-0.75) (c1p) {};
\node[circ] at (3,-2.25) (c2p) {};
\node[circ] at (3,-3.75) (c3p) {};
\node at (3,-4.65) (cdotsp) {$\vdots$};
\node[circ] at (3,-5.75) (cmp) {};

\draw (4,0) -- (4, -6.5);
\draw (4,0) -- (3.5,0);
\draw (4,-6.5) -- (3.5,-6.5);

\node[align=left] at (5.25, -3.375) (bracket) {$m-1$\\non-clause\\vertices};

\draw  (c1) -> (c1p);
\draw  (c1) -> (c2p);
\draw  (c1) -> (c3p);
\draw  (c1) -> (cmp);
\draw  (c2) -> (c1p);
\draw  (c2) -> (c2p);
\draw  (c2) -> (c3p);
\draw  (c2) -> (cmp);
\draw  (c3) -> (c1p);
\draw  (c3) -> (c2p);
\draw  (c3) -> (c3p);
\draw  (c3) -> (cmp);
\draw  (c4) -> (c1p);
\draw  (c4) -> (c2p);
\draw  (c4) -> (c3p);
\draw  (c4) -> (cmp);
\draw  (cm) -> (c1p);
\draw  (cm) -> (c2p);
\draw  (cm) -> (c3p);
\draw  (cm) -> (cmp);
\end{tikzpicture}
\end{center}
\caption{Clause selection gadget for \game{UPR}.}
\label{fig:uprcl}
\end{figure}

\begin{figure}
\begin{center}
\input{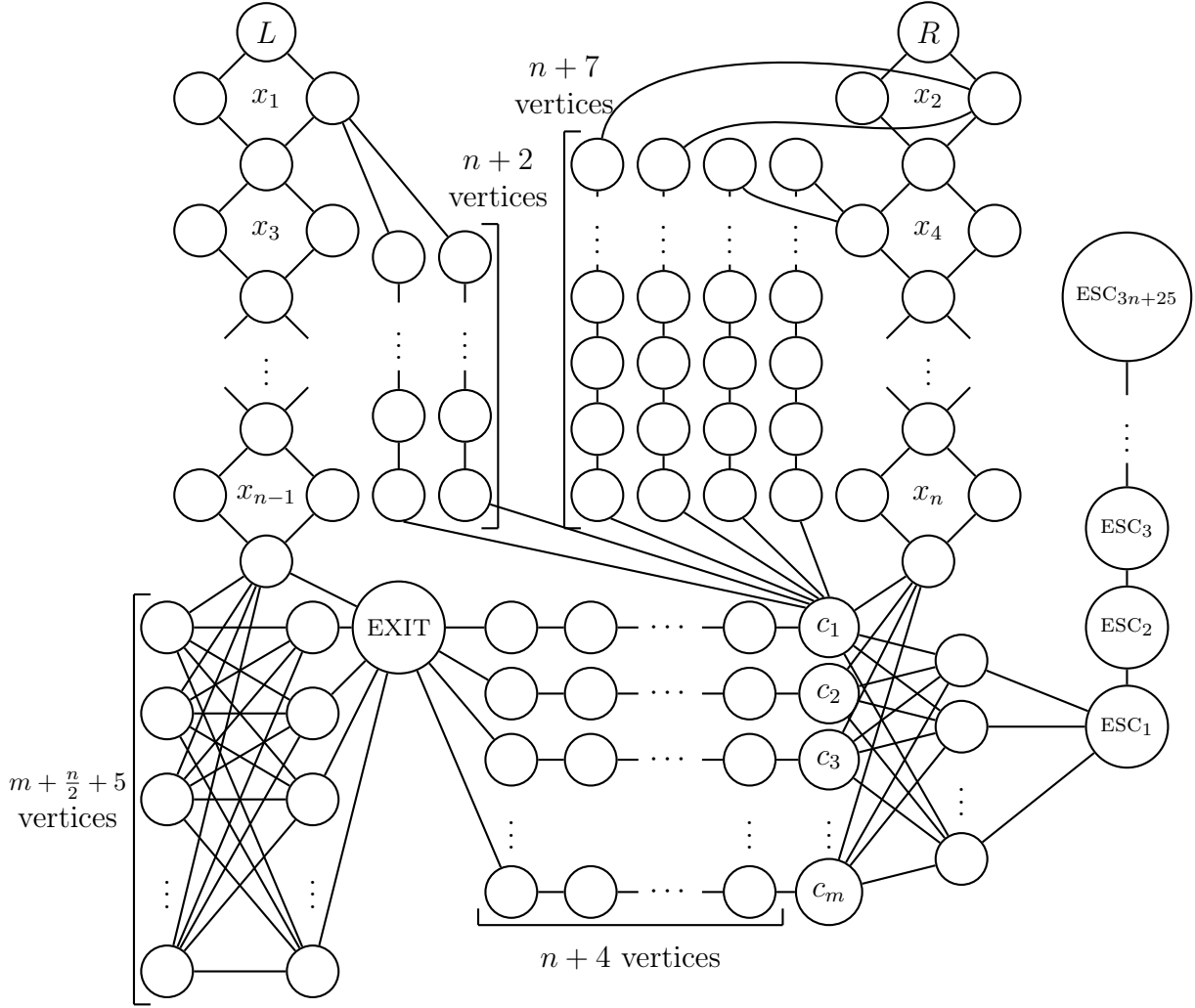}
\end{center}
\caption{Schematic diagram \game{UPR} construction from \game{TQBF} instance $\exists x_1\forall x_2\cdots\exists x_{n-1}\forall x_n\pb{x_1\vee x_2\vee\bar{x}_4}\wedge c_2\wedge\cdots\wedge c_m$. Linkers adjacent to $c_2$ through $c_m$ exist but are not shown.}
\label{fig:upr}
\end{figure}

We claim first that the graph is bipartite. Note that each individual building block is bipartite, as each is a $4$-cycle, a path, or a complete bipartite graph. Consider the following partition of the graph into two parts, $A$ and $B$:
\begin{description}
\item[$A$:]\lb
\begin{itemize}
\item Top and bottom vertices of the gadgets for variables $x_i$ for odd $i$
\item Left and right vertices of the gadgets for variables $x_i$ for even $i$
\item Part~II of the delay graph
\item Every other vertex of each clause connector, starting with the vertex adjacent to \textsc{exit}
\item The clause vertices
\item Every other escape vertex, starting with the first one
\item The vertices in linkers not in $B$.
\end{itemize}
\item[$B$:]\lb
\begin{itemize}
\item Left and right vertices of the gadgets for variables $x_i$ for odd $i$
\item Top and bottom vertices of the gadgets for variables $x_i$ for even $i$
\item Part~I of the delay graph
\item \textsc{exit}
\item The vertices in clause connectors not in $A$
\item The vertices in the clause selection gadget that are not clause vertices
\item Every other escape vertex, starting with the second one
\item Every other vertex in each linker, starting with the vertex adjacent to a clause vertex.
\end{itemize}
\end{description}
We claim that this is a bipartition of the graph. There are only three things to check that aren't immediately clear. First, since $n+4$ is even, the vertex of a clause connector adjacent to a clause vertex is in $B$ (since the vertex adjacent to \textsc{exit} is in $A$). The clause vertices are in $A$, as required. Second, since $n+2$ is even, the vertex in each left-linker adjacent to a variable vertex is in $A$ (since the vertex adjacent to a clause vertex is in $B$). That variable vertex is a left or right vertex of a gadget for $x_i$ with $i$ odd, so it is in $B$, as required. Finally, since $n+7$ is odd, the vertex in each right-linker adjacent to a variable vertex is in $B$ (since the vertex adjacent to a clause vertex is also in $B$). That variable vertex is a left or right vertex of a gadget for $x_i$ with $i$ even, so it is in $A$, as required.

Since Left's token starts in $A$ and Right's token starts in $B$, it is always Left's turn if the tokens are in opposite parts and Right's turn if the tokens are in the same part. For convenience, we refer to this property as the \emph{bipartite move invariant}. A key consequence of the bipartite move invariant is that it is always Left's turn if the tokens are on adjacent vertices.

Here is one possible way the game could progress, which we refer to as \emph{proper play}: First, in what we call \emph{Phase~I}, the tokens move downward through the variable gadgets. At the end of Phase~I, it is Left's turn, with Left's token at the bottom of the gadget for $x_{n-1}$ and Right's token at the bottom of the gadget for $x_n$. Then, to start Phase~II, Left's token enters Part~I of the delay graph immediately before Right's token enters a clause vertex. Left's next $2m-3$ moves take place in the delay graph, while Right's next $2m-3$ moves take place in the clause selection gadget. After all of these moves, Left is at some vertex in Part~II of the delay graph, while Right is at the only remaining non-clause vertex in the clause selection gadget. Moreover, only one clause vertex remains. It is now Left's move, and she moves back to Part~I of the delay graph. Right then moves to the first escape vertex. All of Right's subsequent moves take place in the escape path. Left moves through all remaining vertices in the delay graph and, exhausting all of those, moves to \textsc{exit}. From there, she moves along the clause connector to whichever clause vertex $c_j$ still remains. Upon reaching that clause, she departs via some linker. If she chooses a left-linker and she reaches a vertex in a variable gadget, she continues into another linker attached to that vertex.

We have the following claims, which would complete our reduction:
\begin{itemize}
\item If the game must progress according to proper play, then $Q$ is true if and only if Left wins the game of \game{UPR}.
\item Neither player can do better than to play properly.
\end{itemize}
For the time being, assume that the game must progress properly. First, suppose $Q$ is true. Then, there is some setting of the variables $x_1,x_3,\ldots,x_{n-1}$ so that the formula is true regardless of the values of $x_2,x_4,\ldots,x_n$. Take such a setting, and, as Left, adopt the following strategy:
\begin{itemize}
\item When moving from the top vertex of the gadget for $x_i$, move to the left vertex if $x_i$ should be set to \ttt, and move to the right vertex if $x_i$ should be set to \fff.
\item When moving from a clause vertex, enter a linker corresponding to a true literal.
\end{itemize}
Since $Q$ is true, at least one of that clause's literals is true. 
As a result, the vertex in a variable gadget that that linker connects to the clause vertex is still present. Regardless of whether it is a left-linker or a right-linker, Left has at least $n+8$ moves remaining when her token lies on the clause vertex. If the linker is a left linker, there are $n+2$ moves to vertices along the linker, one move to a vertex in a variable gadget, and five moves into another linker (since $n+2$, the number of vertices in a left linker, is at least $5$). If the linker is a right-linker, there are $n+7$ moves to vertices along the linker and one move to a vertex in a variable gadget. Now, consider the position at the end of Phase~I. Right has $2m+3n+23$ moves remaining: $2m-2$ moves in the clause gadget, followed by $3n+25$ moves in the escape path. Left has at least $2m+3n+24$ moves remaining: $2m+n+10$ moves in the delay graph, $1$ move to \textsc{exit}, $n+4$ moves through a clause connector, $1$ move into a clause vertex, and $n+8$ moves after the clause vertex. Therefore, Left wins, as required.

Now, suppose $Q$ is false. Then, there is some setting of the variables $x_2,x_4,\ldots,x_n$ so that the formula is false regardless of the values of $x_1,x_3,\ldots,x_{n-1}$. Take such a setting, and, as Right, adopt the following strategy:
\begin{itemize}
\item When moving from the top vertex of the gadget for $x_i$, move to the left vertex if $x_i$ should be set to \ttt, and move to the right vertex if $x_i$ should be set to \fff.
\item For some $j$ such that clause $c_j$ is false, avoid moving to the vertex corresponding to $c_j$ while traversing the clause selection gadget.
\end{itemize}
Around the end of the game, Left moves to the vertex for $c_j$. Since $c_j$ is false, all of that clause's literals are false. This means that Left has at most $n+7$ moves remaining from the clause vertex, since each linker from $c_j$ no longer connects to a variable gadget. Now, consider the position at the end of Phase~I. Right has $2m+3n+23$ moves remaining, as we saw in the previous case. Left has at most $2m+3n+23$ moves remaining: $2m+n+10$ moves in the delay graph, $1$ move to \textsc{exit}, $n+4$ moves through a clause connector, $1$ move into a clause vertex, and $n+7$ moves after the clause vertex. Therefore, Right wins, as required.

What remains to show is that neither player can do better by deviating from proper play. Each player has several ways in which they can deviate. We analyze each possible deviation one at a time, supposing that it is the first deviation that has taken place in the game thus far.
\begin{description}
\item[Left enters a linker from a variable gadget:] Suppose Left enters a linker connected to the clause vertex for $c_j$. When Left enters a linker, Right's token is on the left or right vertex of a variable gadget. Such a vertex is at most $n$ moves from the clause vertex for $c_j$ if Right continues down through the variable gadgets. (At worst, Right is in the $x_2$ gadget. There are $\frac{n}{2}$ gadgets, each requiring two moves to traverse. One move of the first gadget is already complete, and there is one additional move required to reach a clause.) Since it requires $n+3$ moves for Left to reach the vertex for $c_j$ via the linker, Right can reach $c_j$ before Left, deleting it as he moves elsewhere, forcing Left to run out of moves at the end of the linker.
\item[Left moves upward from the bottom vertex of a variable gadget:] Moving upward in a variable gadget forces the subsequent move to be into a linker. Right can proceed as in the previous case, now having even more time to reach the appropriate clause vertex before Left. There is one small additional consideration: If Left moves upward from the bottom of the $x_{n-1}$ gadget, Right does not know which clause she will head for and must now move to a clause himself. But, he can move to any clause, after which Left commits to a clause. Then, Right takes two more moves to reach that clause. From there, Right can move into the clause connector, and Left runs out of moves first, since left-linkers are shorter than clause connectors.
\item[Left moves to \textsc{exit} before visiting the whole delay graph:] From \textsc{exit}, Left can either move back into the delay graph or into a clause connector. If she returns to the delay graph, she is trapped there, running out of moves while Right plays properly. So, we henceforth assume that Left moves to \textsc{exit} with the intent to move next into a clause connector. Now, observe that, if Right has not yet entered the escape path when Left moves to \textsc{exit}, the bipartite move invariant implies that Right's token is not on a clause vertex; it must be in the same part ($B$) of the underlying graph as \textsc{exit}. There are two sub-cases to consider regarding the situation when Left moves to \textsc{exit}.
\begin{description}
\item[At least two clause vertices remain:] After Left moves to \textsc{exit}, Right moves to any clause vertex. Then, Left definitely loses if she moves toward the clause vertex Right's token currently occupies or toward a deleted clause vertex, so she moves toward a remaining clause vertex. Right then moves to a non-clause vertex in the clause selection gadget and then to the clause vertex that Left is headed for. From there, Right enters the same clause connector that Left is currently traversing. By the bipartite move invariant, it is Left's move when the tokens are adjacent. This results in Left losing, as required.
\item[Exactly one clause vertex remains:] If the remaining clause is true, Left would win anyway, so we can assume that the remaining clause is false and that Left is trying to win when $Q$ is false. From this point forward, Left's only option is to follow proper play, running out of moves sooner than if she had not left the delay graph early. Therefore, if Right continues according to proper play, Left loses.
\end{description}
\item[Right enters a linker from a variable gadget:] Suppose Right enters a linker connected to the clause vertex for $c_j$. When Right enters a linker, Left's token is at the bottom of a variable gadget. Such a vertex is at most $n+6$ moves from the clause vertex for $c_j$, if Left moves up the current variable gadget (one move), through a linker to a clause vertex ($n+2$ moves in the linker plus one to the clause vertex), and if necessary, through the clause selection gadget to $c_j$ ($2$ moves). Since it requires $n+7$ additional moves for Right to reach the vertex for $c_j$ via the right-linker after entering it, Left can reach the vertex for $c_j$ before Right, causing Right to run out of moves at the end of the linker.
\item[Right moves upward from the bottom vertex of a variable gadget:] Moving upward in a variable diamond forces the subsequent move to be into a linker. If Left's token is still in a variable gadget, it is on a left or right vertex, and Left can proceed as in the previous case, now having even more time to reach the appropriate clause vertex before Right. The other possibility is that Left's token is in Part~I of the delay graph. In this case, Left's next move is to any vertex in Part~II, after which Right commits to a clause $c_j$. Then, Left moves to \textsc{exit} and heads for the vertex for $c_j$ via the clause connector. The total number of moves required by Left to reach the vertex for $c_j$ once Right deviates from proper play is $n+7$ (one in the delay graph, one to \textsc{exit}, $n+4$ through the clause connector, and one to $c_j$), and Right needs $n+8$ moves to reach $c_j$ \emph{after} deviating ($n+7$ through the linker and one to the vertex for $c_j$). Therefore, Left reaches the vertex for $c_j$ before Right and wins.
\item[Right moves into a linker from a clause vertex:] We can assume without loss of generality that the variable vertex at the end of the linker Right enters is still present in the graph, as otherwise Left can continue moving properly while Right runs out of moves. This means that Right reaches that variable vertex and then immediately enters another linker. By the bipartite move invariant, when Right moves into a linker (from $A$ to $B$), Left's token is in Part~II of the delay graph (in $A$). There are two sub-cases to consider:
\begin{description}
\item[Right enters a left-linker:] In this case, Left can immediately move to \textsc{exit}, and then she can enter any clause connector toward a remaining clause. It takes $n+6$ moves for Left to reach a clause. After $n+6$ moves, since Right originally entered a left-linker, he has already passed through the variable vertex and has therefore committed to a linker to return along back to the clause selection gadget. If the clause vertex at the end of that linker still remains, Left can move to it through the clause selection gadget in two more moves, and then she can exit through the clause connector, after which Right runs out of moves first. It takes Left $n+8$ moves to reach this second clause vertex, and it takes Right $2n+6$ moves to reach it, with Right moving first. Since $n\geq4$, we have $n+8<2n+6$, so Left does indeed reach the clause vertex first, and since $n+2<n+4$, Right does indeed run out of moves first.
\item[Right enters a right-linker:] In this case, Left should move back to Part~I of the delay graph, then back to Part~II, then back to Part~I, then back to Part~II, then to \textsc{exit}. Then, she can enter any clause connector toward a remaining clause. It takes $n+10$ moves for Left to reach a clause this way. After $n+10$ moves, Right's token has entered another linker after visiting a variable vertex, and Right moves next and reaches a clause vertex (if it is still present) in $n+6$ more moves. If there is no clause that Right is headed toward, Left heads for the escape path and wins. If there is a clause vertex that Right is headed toward, Left reaches that clause in two moves and then enters the clause connector. This is a total of $n+6$ more moves for Left (two in the clause selection gadget and $n+4$ through the clause connector). Since Right moves first from this checkpoint, Right runs out of moves first and loses.
\end{description}
\item[Right moves into a clause connector from a clause vertex:] By the bipartite move invariant, when Right moves into a clause connector (into $B$), Left's token is in Part~II of the delay graph (in $A$). Left can immediately move to \textsc{exit}, and then she can enter any other clause connector. Right runs out of moves first, when he reaches the vertex in his clause connector that was once adjacent to \textsc{exit}.
\item[Right moves into the last clause vertex instead of the escape path:] This forces Right's next move to be into a clause connector or a linker. If Right moves into a clause connector, play continues according to that case above, with Left winning. If Right moves into a linker, right has at most $2n+14$ moves remaining after doing so, as there is no clause remaining to return to along a second linker. So, after Phase~I ends, Right has at most $2m+2n+14$ remaining moves, accounting for all $2m-1$ vertices in the clause gadget and $2n+15$ moves afterward. At this same point, Left has $2m+n+10$ moves in the delay graph, $1$ move to \textsc{exit}, and $n+4$ moves through a clause connector. This is a total of $2m+2n+15$ moves, so Right runs out of move first and loses.
\item[Right moves into the escape path with multiple remaining clause vertices:] If all remaining clauses are false, Right would win by playing properly, so we can assume that at least one clause is true and that Right is trying to win when $Q$ is true. Once Right enters the escape path, he can never leave it, and he has exactly $2n+24$ moves afterward. This means that Left can now continue properly, exhausting the delay graph, then traveling to a remaining true clause vertex and into a linker corresponding to a true literal. Right has fewer moves than before, so Right still runs out of moves before Left.
\end{description}

This completes the reduction showing that $\text{\game{TQBF}}\leq_p \game{UPR}$. Therefore, \game{UPR} is \ps-hard on bipartite graphs, and, hence, \ps-complete, as required.
\end{proof}


\subsection{Complexity of UPF}\label{ss:upf}

We now prove that the related game \game{UPF} is also \ps-complete. In this case, our construction is not bipartite\footnote{This construction depends critically on odd cycles (here $13$-cycles) resulting from the way variable gadgets are attached to one another.}.

\begin{theorem}\label{thm:upf}
The game \game{UPF} is \ps-complete.
\end{theorem}
\begin{proof}
We proceed via a reduction from \game{TQBF} to \game{UPF}. Let $Q=\exists x_1\forall x_2\cdots\exists x_{n-1}\forall x_n c_1\wedge c_2\wedge\cdots\wedge c_m$ be an instance of \game{TQBF}. Without loss of generality, assume that $m\geq3$. 
We construct an instance of \game{UPF} as follows:
\begin{itemize}
\item For each variable $x_i$, we have what we call the \emph{variable gadget,} which consists of two disjoint pieces: a cycle with $8$ vertices and a path with $5$ vertices. We draw and refer to the cycle as a \emph{diamond}, allowing us to define its top, leftmost, rightmost, and bottom vertices.  We draw the path vertically with the first vertex on top. See Figure~\ref{fig:upfvar} for a depiction of a variable gadget.
\item Left's token starts at the top vertex of the $x_1$ diamond. Right's token starts at the first vertex of the $x_1$ path.
\item For $1\leq i<n$, the bottom vertex of the $x_i$ diamond is adjacent to the first vertex of the $x_{i+1}$ path, and the fourth vertex of the $x_i$ path is adjacent to the top vertex of the $x_{i+1}$ diamond.
\item For $1\leq i\leq n$, the bottom vertex of the $x_i$ diamond is adjacent to the fifth vertex of the $x_i$ path.
\item There is a graph called the \emph{clause selection gadget}, which consists of several pieces. See Figure~\ref{fig:upfcl} for a depiction of the clause selection gadget; what follows are descriptions of the pieces and how they are interconnected.
\begin{itemize}
\item There are designated top and bottom vertices. The top vertex of the clause selection gadget is also the bottom vertex of the $x_n$ diamond.
\item For each clause $c_j$, there is a single vertex.
\item There are $m$ paths enumerated from $1$ to $m$, called \emph{clause deletion paths}, of $m+3$ vertices each. The first vertex of each clause deletion path is adjacent to the top vertex of the clause selection gadget. The last vertex of each clause deletion path is adjacent to the bottom vertex of the clause selection gadget.
\item Each vertex in each clause deletion path other than the first, second, second-to-last, and last, is adjacent to a clause vertex. The vertex for clause $c_j$ is adjacent to some vertex in every clause deletion path except for the $j^{th}$ path. Along a path, the indices of the adjacent clause vertices appear in increasing order.
\end{itemize}
\item There is a path of $m+9$ vertices called the \emph{delay path}. The first delay vertex is adjacent to the fourth vertex in the path for $x_n$.
\item There is a path linking the last delay vertex to each clause vertex. There are $m$ vertices along this path, not including the delay vertex and the clause vertex. Each such path is called a \emph{clause connector}.
\item For each literal $x_i$ appearing in clause $c_j$, there is a path from the vertex for $c_j$ to the right vertex in the diamond for $x_i$. For each literal $\bar{x}_i$ appearing in clause $c_j$, there is a path from the vertex for $c_j$ to the left vertex in the diamond for $x_i$. There are $m+3$ vertices along each path, not including the clause vertex and the variable vertex. We refer to each such path linking a clause vertex back to a variable gadget as a \emph{linker}.
\item There is a path of $8m+7$ vertices, known as the \emph{escape path}. The first escape vertex is adjacent to the bottom vertex of the clause selection gadget.
\item There is a vertex called \textsc{prize} adjacent to the third escape vertex and the $\pb{m+7}^{th}$ delay vertex.
\item There is a path, called the \emph{win path}, of $13n+5m^2+22m+21$ vertices. The first win vertex is adjacent to \textsc{prize} and to the fifth vertex of each path in each variable gadget.
\end{itemize}
Observe that the number of vertices in the win path is two greater than the number of vertices in the graph outside the win path. (There are $13n$ vertices in variable gadgets, $2+m+m\pb{m+3}$ in the clause gadget, $m+9$ delay vertices, $m\pb{m+3}$ total in the clause connectors, $3m^2$ total in the linkers, $8m+7$ escape vertices, and $1$ \textsc{prize} vertex.)

See Figure~\ref{fig:upf} for a schematic of this construction. There are several linkers that are not pictured; only the linkers from $c_1$ are shown. Also, there are edges between clause vertices and omitted vertices in clause deletion paths, and there are edges from omitted variable gadgets to the first win vertex.

\begin{figure}
\begin{center}
\begin{tikzpicture}[thick, node distance=1.5cm, circ/.style={draw, circle, minimum size=20pt}]

\node[circ] at (0, 0) (top) {};
\node[circ] at (-1cm, -1cm) (toptrue) {};
\node[circ] at (1cm, -1cm) (topfalse) {};
\node[circ] at (-2cm, -2cm) (true) {$1$};
\node[circ] at (2cm, -2cm) (false) {$0$};
\node[circ] at (-1cm, -3cm) (bottomtrue) {};
\node[circ] at (1cm, -3cm) (bottomfalse) {};
\node[circ] at (0, -4cm) (bottom) {};

\node[circ] at (3cm, 0) (p1) {};
\node[circ] at (3cm, -1cm) (p2) {};
\node[circ] at (3cm, -2cm) (p3) {};
\node[circ] at (3cm, -3cm) (p4) {};
\node[circ] at (3cm, -4cm) (p5) {};

\node (ind) at (0, 0.8cm) {};
\node (outd) at (0, -4.8cm) {};

\node (inp) at (3cm, 0.8cm) {};
\node (outp) at (3.8cm, -4.8cm) {};

\node (outw) at  (4.2cm, -2.2cm) {};

\node at (0, -2cm) (var) {$x_i$};

\draw (top) to (toptrue);
\draw (toptrue) to (true);
\draw (true) to (bottomtrue);
\draw (bottomtrue) to (bottom);
\draw (top) to (topfalse);
\draw (topfalse) to (false);
\draw (false) to (bottomfalse);
\draw (bottomfalse) to (bottom);
\draw (p1) to (p2);
\draw (p2) to (p3);
\draw (p3) to (p4);
\draw (p4) to (p5);
\draw (bottom) to (p5);
\draw (ind) to (top);
\draw (bottom) to (outd);
\draw (inp) to (p1);
\draw (p4) to [out=350, in=90] (outp);
\draw (p5) to (outw);
\end{tikzpicture}
\end{center}
\caption{Variable gadget for \game{UPF}. Symbols $1$ and $0$ denote vertices corresponding to the variable $x_i$ being set to {\ttt} or \fff, respectively.}
\label{fig:upfvar}
\end{figure}
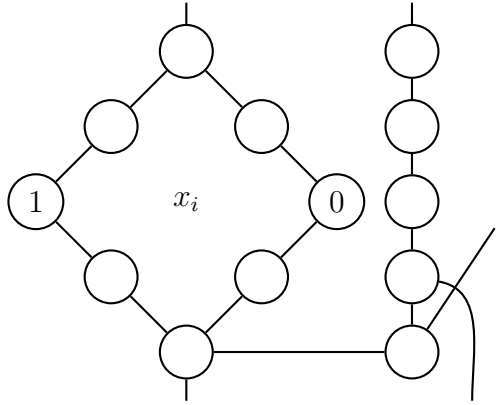

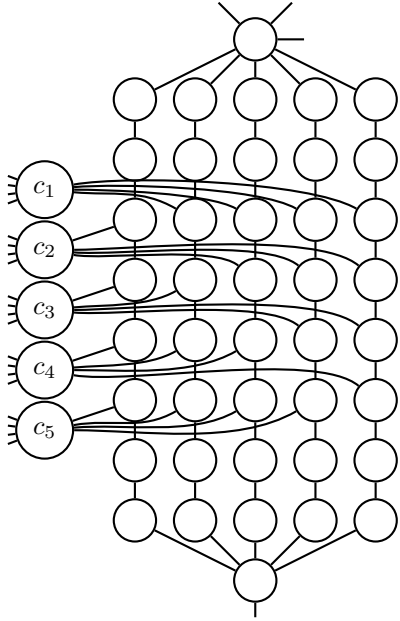
\begin{figure}
\begin{center}
\begin{tikzpicture}[thick, node distance=1.5cm, circ/.style={draw, circle, minimum size=16pt}, scale=0.8]

\node[circ] at (0, 0) (top) {};
\node[circ] at (-2cm, -1cm) (c11) {};
\node[circ] at (-2cm, -2cm) (c12) {};
\node[circ] at (-2cm, -3cm) (c13) {};
\node[circ] at (-2cm, -4cm) (c14) {};
\node[circ] at (-2cm, -5cm) (c15) {};
\node[circ] at (-2cm, -6cm) (c16) {};
\node[circ] at (-2cm, -7cm) (c17) {};
\node[circ] at (-2cm, -8cm) (c18) {};
\node[circ] at (-1cm, -1cm) (c21) {};
\node[circ] at (-1cm, -2cm) (c22) {};
\node[circ] at (-1cm, -3cm) (c23) {};
\node[circ] at (-1cm, -4cm) (c24) {};
\node[circ] at (-1cm, -5cm) (c25) {};
\node[circ] at (-1cm, -6cm) (c26) {};
\node[circ] at (-1cm, -7cm) (c27) {};
\node[circ] at (-1cm, -8cm) (c28) {};
\node[circ] at (0cm, -1cm) (c31) {};
\node[circ] at (0cm, -2cm) (c32) {};
\node[circ] at (0cm, -3cm) (c33) {};
\node[circ] at (0cm, -4cm) (c34) {};
\node[circ] at (0cm, -5cm) (c35) {};
\node[circ] at (0cm, -6cm) (c36) {};
\node[circ] at (0cm, -7cm) (c37) {};
\node[circ] at (0cm, -8cm) (c38) {};
\node[circ] at (1cm, -1cm) (c41) {};
\node[circ] at (1cm, -2cm) (c42) {};
\node[circ] at (1cm, -3cm) (c43) {};
\node[circ] at (1cm, -4cm) (c44) {};
\node[circ] at (1cm, -5cm) (c45) {};
\node[circ] at (1cm, -6cm) (c46) {};
\node[circ] at (1cm, -7cm) (c47) {};
\node[circ] at (1cm, -8cm) (c48) {};
\node[circ] at (2cm, -1cm) (c51) {};
\node[circ] at (2cm, -2cm) (c52) {};
\node[circ] at (2cm, -3cm) (c53) {};
\node[circ] at (2cm, -4cm) (c54) {};
\node[circ] at (2cm, -5cm) (c55) {};
\node[circ] at (2cm, -6cm) (c56) {};
\node[circ] at (2cm, -7cm) (c57) {};
\node[circ] at (2cm, -8cm) (c58) {};
\node[circ] at (0, -9cm) (bottom) {};

\node[circ] at (-3.5cm, -2.5cm) (c1) {{\footnotesize $c_1$}};
\node[circ] at (-3.5cm, -3.5cm) (c2) {{\footnotesize $c_2$}};
\node[circ] at (-3.5cm, -4.5cm) (c3) {{\footnotesize $c_3$}};
\node[circ] at (-3.5cm, -5.5cm) (c4) {{\footnotesize $c_4$}};
\node[circ] at (-3.5cm, -6.5cm) (c5) {{\footnotesize $c_5$}};

\node (in1) at (-0.8cm, 0.8cm) {};
\node (in2) at (0.8cm, 0.8cm) {};
\node (in3) at (1cm, 0cm) {};
\node (out) at (0, -9.8cm) {};

\node (c1out1) at (-4.3cm, -2.2cm) {};
\node (c1out2) at (-4.3cm, -2.4cm) {};
\node (c1out3) at (-4.3cm, -2.6cm) {};
\node (c1out4) at (-4.3cm, -2.8cm) {};
\node (c2out1) at (-4.3cm, -3.2cm) {};
\node (c2out2) at (-4.3cm, -3.4cm) {};
\node (c2out3) at (-4.3cm, -3.6cm) {};
\node (c2out4) at (-4.3cm, -3.8cm) {};
\node (c3out1) at (-4.3cm, -4.2cm) {};
\node (c3out2) at (-4.3cm, -4.4cm) {};
\node (c3out3) at (-4.3cm, -4.6cm) {};
\node (c3out4) at (-4.3cm, -4.8cm) {};
\node (c4out1) at (-4.3cm, -5.2cm) {};
\node (c4out2) at (-4.3cm, -5.4cm) {};
\node (c4out3) at (-4.3cm, -5.6cm) {};
\node (c4out4) at (-4.3cm, -5.8cm) {};
\node (c5out1) at (-4.3cm, -6.2cm) {};
\node (c5out2) at (-4.3cm, -6.4cm) {};
\node (c5out3) at (-4.3cm, -6.6cm) {};
\node (c5out4) at (-4.3cm, -6.8cm) {};

\draw (top) to (c11);
\draw (c11) to (c12);
\draw (c12) to (c13);
\draw (c13) to (c14);
\draw (c14) to (c15);
\draw (c15) to (c16);
\draw (c16) to (c17);
\draw (c17) to (c18);
\draw (c18) to (bottom);
\draw (top) to (c21);
\draw (c21) to (c22);
\draw (c22) to (c23);
\draw (c23) to (c24);
\draw (c24) to (c25);
\draw (c25) to (c26);
\draw (c26) to (c27);
\draw (c27) to (c28);
\draw (c28) to (bottom);
\draw (top) to (c31);
\draw (c31) to (c32);
\draw (c32) to (c33);
\draw (c33) to (c34);
\draw (c34) to (c35);
\draw (c35) to (c36);
\draw (c36) to (c37);
\draw (c37) to (c38);
\draw (c38) to (bottom);
\draw (top) to (c41);
\draw (c41) to (c42);
\draw (c42) to (c43);
\draw (c43) to (c44);
\draw (c44) to (c45);
\draw (c45) to (c46);
\draw (c46) to (c47);
\draw (c47) to (c48);
\draw (c48) to (bottom);
\draw (top) to (c51);
\draw (c51) to (c52);
\draw (c52) to (c53);
\draw (c53) to (c54);
\draw (c54) to (c55);
\draw (c55) to (c56);
\draw (c56) to (c57);
\draw (c57) to (c58);
\draw (c58) to (bottom);

\draw (c1) to [out=355, in=150] (c23);
\draw (c1) to [out=0, in=150, looseness=0.5] (c33);
\draw (c1) to [out=5, in=150, looseness=0.5] (c43);
\draw (c1) to [out=10, in=150, looseness=0.5] (c53);
\draw (c2) to (c13);
\draw (c2) to [out=350, in=140, looseness=0.5] (c34);
\draw (c2) to [out=355, in=140, looseness=0.5] (c44);
\draw (c2) to [out=0, in=140, looseness=0.5] (c54);
\draw (c3) to (c14);
\draw (c3) to [out=5, in=210, looseness=0.8] (c24);
\draw (c3) to [out=355, in=140, looseness=0.5] (c45);
\draw (c3) to [out=0, in=140, looseness=0.5] (c55);
\draw (c4) to (c15);
\draw (c4) to [out=5, in=210, looseness=0.8] (c25);
\draw (c4) to [out=0, in=210, looseness=0.8] (c35);
\draw (c4) to [out=350, in=140, looseness=0.5](c56);
\draw (c5) to (c16);
\draw (c5) to [out=10, in=210, looseness=1] (c26);
\draw (c5) to [out=5, in=210, looseness=0.8] (c36);
\draw (c5) to [out=0, in=210, looseness=0.8] (c46);

\draw (in1) to (top);
\draw (in2) to (top);
\draw (in3) to (top);
\draw (bottom) to (out);

\draw (c1) to (c1out1);
\draw (c1) to (c1out2);
\draw (c1) to (c1out3);
\draw (c1) to (c1out4);
\draw (c2) to (c2out1);
\draw (c2) to (c2out2);
\draw (c2) to (c2out3);
\draw (c2) to (c2out4);
\draw (c3) to (c3out1);
\draw (c3) to (c3out2);
\draw (c3) to (c3out3);
\draw (c3) to (c3out4);
\draw (c4) to (c4out1);
\draw (c4) to (c4out2);
\draw (c4) to (c4out3);
\draw (c4) to (c4out4);
\draw (c5) to (c5out1);
\draw (c5) to (c5out2);
\draw (c5) to (c5out3);
\draw (c5) to (c5out4);
\end{tikzpicture}
\end{center}
\caption{Clause selection gadget for \game{UPF} with $m=5$}
\label{fig:upfcl}
\end{figure}

\begin{figure}
\begin{center}
\input{UPF.tex}
\end{center}
\caption{Schematic diagram \game{UPF} construction from \game{TQBF} instance $\exists x_1\forall x_2\cdots\exists x_{n-1}\forall x_n\pb{x_1\vee x_2\vee\bar{x}_n}\wedge c_2\wedge\cdots\wedge c_m$. Linkers adjacent to $c_2$ through $c_m$ exist but are not shown.}
\label{fig:upf}
\end{figure}

Here is one possible way the game could progress, which we refer to as \emph{proper play}. We split the progression into three phases:
\begin{description}
\item[Phase I:] This phase consists of one round for each variable gadget, along with a move by each player between each pair of successive rounds. The rounds for odd-indexed variables are begun by Left; the rounds for even-indexed variables are begun by Right. In a single round, four moves are made by the first player and three by the second. Whichever player begins spends their moves moving from the top vertex of the diamond to the bottom vertex of the diamond. They delete the top vertex on the first move and the left or right vertex on the third move (whichever direction they choose to travel). The second move can either delete the vertex moved from or a vertex in a linker. The fourth move must delete the fifth vertex in the path for this variable if it is still present. Otherwise, it can delete either neighbor in the current diamond. The other player spends their three moves moving regularly along their path, except that the third move is also permitted to delete the fifth vertex in the path.

Between successive rounds, the player in the path, currently on its fourth vertex, makes a regular move to the top vertex of the next variable's diamond, and then the player in the diamond, currently in its bottom vertex, moves regularly to the first vertex of the next variable's path.
\item[Phase II:] This phase begins after the final round of Phase I, with Left's token on the fourth vertex of the $x_n$ path, Right's vertex on the bottom vertex of the $x_n$ diamond, and Left to move. It consists of $m+7$ moves by each player. All of Left's moves in Phase II are regular moves through the delay path. Right's first $m+4$ moves are through the clause selection gadget, along one of the clause deletion paths. On each move that ends adjacent to a clause vertex, Right deletes that clause vertex, resulting in the deletion of all but one clause vertex. The rest of Right's moves through the clause selection gadget are regular, with the possible exception of the move to the bottom of the gadget, which can delete any neighbor in the clause selection gadget. Right's subsequent three moves are into and along the escape path. The first two of these moves are regular; the third deletes \textsc{prize}.
\item[Phase III:] This phase begins with Left in the delay path on the vertex that was adjacent to \textsc{prize} before Right deleted it and with Right in the escape path on the vertex that was adjacent to \textsc{prize}. It is Left's turn. Here is the sequence of moves Left makes:
\begin{itemize}
\item First, she continues to move regularly through the delay path.
\item Next, she enters the \emph{clause connector stalling stage}. There are $m-1$ clause connectors that are no longer connected to a clause. She moves into one of them, deleting the subsequent vertex in that connector, then she returns regularly to the last delay vertex. She does this two move sequence for each of these $m-1$ clause connectors.
\item Once those moves are all done, she enters and moves regularly along the clause connector connected to the one remaining clause vertex until she reaches the clause vertex.
\item Now, Left enters the \emph{clause deletion path stalling stage}. The clause vertex is adjacent to $m-1$ vertices in $m-1$ clause deletion paths, which have no connections remaining between them except through the clause vertex. She moves into one of the clause deletion paths, deleting the vertex in that path below the vertex moved to. Then, she moves upward along the path, deleting the vertex above the one moved to. 
She then takes two regular moves to return to the clause vertex. She does this four move sequence for each of the $m-1$ clause deletion paths. See Figure~\ref{fig:upfcdpdel} for an illustration of this four move sequence.
\item Once those moves are all done, she enters the \emph{linker stalling stage}. The clause vertex is adjacent to three linkers. She moves into one of them, deleting the second vertex, then she returns regularly to the clause vertex. Left then repeats this process for a second linker.
\item Finally, Left enters and moves along the third linker, potentially reaching a variable diamond.
\end{itemize}
Meanwhile, Right's moves are all regular along the escape path.
\end{description}
Note that, in Phase~III, Left makes two moves in the delay path, followed by $2m-2$ moves in the clause connector stalling stage, followed by $m$ moves through the last clause connector and one move to the clause vertex, followed by $4m-4$ moves in the clause deletion path stalling stage, followed by four moves in the linker stalling stage. So, when Left is on the clause vertex about to enter the last linker, she has made a total of $7m+1$ moves (and therefore Right has made the same number of moves).

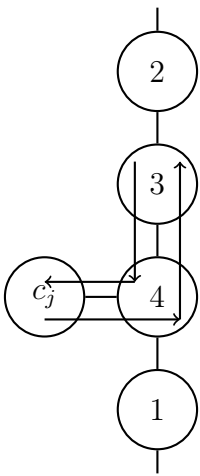
\begin{figure}
\begin{center}
\begin{tikzpicture}[thick, node distance=1.5cm, circ/.style={draw, circle, minimum size=30pt}]

\node[circ, align=center] at (0, 0) (clause) {$c_j$};
\node[circ, align=center] at (1.5cm, 0) (neighbor) {4};
\node[circ, align=center] at (1.5cm, -1.5cm) (ndown) {1};
\node[circ, align=center] at (1.5cm, 1.5cm) (nup) {3};
\node[circ, align=center] at (1.5cm, 3cm) (nupup) {2};
\node at (1.5cm, 4cm) (nupupup) {};
\node at (1.5, -2.5cm) (ndowndown) {};

\draw (clause) to (neighbor);
\draw (neighbor) to (ndown);
\draw (neighbor) to (nup);
\draw (nup) to (nupup);
\draw (nupup) to (nupupup);
\draw (ndown) to (ndowndown);

\draw [->] (0, -0.3) to (1.8, -0.3);
\draw [->] (1.8, -0.3) to (1.8, 1.8);
\draw [->] (1.2, 1.8) to (1.2, 0.2);
\draw [->] (1.2, 0.2) to (0, 0.2);
\end{tikzpicture}
\end{center}
\caption{Illustration of the clause deletion path stalling stage of Phase~III. Arrows indicate moves; numbers indicate order vertices are deleted.}
\label{fig:upfcdpdel}
\end{figure}

We have the following claims, which would complete our reduction:
\begin{itemize}
\item If the game must progress properly, then $Q$ is true if and only if Left wins the game of \game{UPF}.
\item Neither player can do better than to move properly.
\end{itemize}
For the time being, assume that the game must progress properly. First, suppose $Q$ is true. Then, there is some setting of the variables $x_1,x_3,\ldots,x_{n-1}$ so that the formula is true regardless of the values of $x_2,x_4,\ldots,x_n$. Take such a setting, and, as Left, adopt the following strategy:
\begin{itemize}
\item When moving from the top vertex of the diamond for $x_i$, move left if $x_i$ should be set to \ttt, and move right if $x_i$ should be set to \fff.
\item When entering the third linker from a clause vertex, enter a linker corresponding to a true literal.
\end{itemize}
Since $Q$ is true, at least one of that clause's literals is true. 
As a result, the vertex in a variable gadget that that linker connects to the clause vertex is still present. So, Left has at least $8m+5$ moves remaining at the end of Phase~II. Meanwhile, Right has exactly $8m+4$ moves remaining, as he begins Phase~III on the third vertex in the escape path. Therefore, Left wins, as required.

Now, suppose $Q$ is false. Then, there is some setting of the variables $x_2,x_4,\ldots,x_n$ so that the formula is false regardless of the values of $x_1,x_3,\ldots,x_{n-1}$. Take such a setting, and, as Right, adopt the following strategy:
\begin{itemize}
\item When moving from the top vertex of the diamond for $x_i$, move left if $x_i$ should be set to \ttt, and move right if $x_i$ should be set to \fff.
\item For some $j$ such that clause $c_j$ is false, traverse the clause selection gadget along the $j^{th}$ clause deletion path.
\end{itemize}
Around the end of the game, Left moves to the vertex for $c_j$. Since $c_j$ is false, all of that clause's literals are false. This means that Left has at most $8m+4$ moves available in Phase~III, since each linker from $c_j$ no longer connects to a variable gadget. Right has exactly $8m+4$ moves remaining, and Left moves first. Therefore, Right wins, as required.

What remains to show is that neither player can do better if they deviate from proper play. Each player has several ways in which they can deviate. A common, silly deviation that either player can do at essentially any time is to delete the vertex they would otherwise move to next. Usually, this causes that player to be trapped after one more move. So, we ignore that possibility if it would immediately trap them, and we analyze each other possible deviation.. A key observation is that due to the long length of the win path, any player entering its first vertex definitely wins the game. In each item that follows, suppose that it represents the first deviation by either player in the whole game.
\begin{description}
\item[A variable diamond player deviates in Phase~I:] Any deviation by the variable diamond player (deleting the wrong vertex, moving into a linker, backtracking after deleting a linker vertex) either results in the diamond player immediately being trapped or delays that player's arrival at the bottom vertex of the variable diamond. This allows the path player to reach the fifth vertex of the path and subsequently enter the win path. None of the diamond player's deviations can prevent this victory by the path player.
\item[Right does not move properly in Phase~II:] (This case refers to move deviation, not deletion deviations.) Left reaches \textsc{prize} in $m+8$ moves and subsequently enters the win path and wins, unless Right can delete \textsc{prize} fast enough or if he can otherwise enter the win path sooner. Since the game is in Phase~II and the win path was connected only to the fifth vertices of the variable paths and to \textsc{prize}, the only remaining entry point to the win path is through \textsc{prize}. First, observe that the bottom of the clause selection gadget cannot be reached any faster than $m+4$ moves. Since the clause deletion paths always allow deletion of clauses in ascending order, moving through a clause to jump from one clause deletion path to another can only allow advancing one extra move along that path (at the cost of two moves). Furthermore, the linkers are long enough that they cannot be used to shortcut travel past a large number of clauses. The other way Right could try to race Left to the win path is by traveling to $c_1$ in, say, the last clause deletion path and then through a clause connector and into the delay path. To reach the delay vertex adjacent to \textsc{prize} would take $m+7$ moves: three into the clause deletion path, one to $c_1$, $m$ through the clause connector, and three through the delay path. But, it takes Left $m+7$ moves in Phase~II to reach that same delay vertex, and Left moves first in Phase~II. Therefore, this deviation by Right allows Left to reach \textsc{prize} and, hence, the win path.
\item[Right does not delete properly in Phase~II:] We have just seen that Right must move properly, so Right should never delete the vertex they would otherwise move to. 
The only nontrivial way Right can deviate is by making a regular move along a clause deletion path instead of deleting a clause vertex. Left can still play Phase~III properly, but she might now also have additional moves available. Therefore, this deviation by Right cannot cause him to win an otherwise lost game.
\item[Left deletes \textsc{prize}:] This deviation does not affect the outcome of the game. At the point when Left is able to delete \textsc{prize} at the end of Phase~II, Right is on the second escape vertex with all of his remaining moves forced. Left's only sensible subsequent move is to continue to move along the delay path, and from there no incentive exists to deviate again.
\item[Left deviates in Phase~III:] Since Left is the only player who can make meaningful choices in Phase~III, it suffices only to show that Left cannot win a game where $Q$ is false. In other words, we must show that Left can make no more than $8m+4$ moves in Phase~III if $Q$ is false. Here are all of the nontrivial ways Left can deviate:
\begin{description}
\item[Skipping one or more clause connector stalls:] Each clause connector skipped in the clause connector stalling stage reduces Left's total number of moves available by two, so shortening this stage cannot be beneficial.
\item[Intermingling/reordering the later stalling phases:] Left can perform the linker stalling sequences and the clause deletion path stalling sequences in whatever order she pleases without increasing or decreasing the total number of moves.
\item[Doing a third linker stall and ending up in a clause deletion path:] First of all, this requires replacing a clause deletion path stalling sequence (four moves) with a linker stalling sequence (two moves). Then, there are at most $m+3$ moves available in the clause deletion path (including the move used to enter it), since that is the number of vertices in a clause deletion path. Since there are also $m+3$ vertices in a linker, this reduces Left's total moves available by at least two.
\end{description}
\end{description}

This completes the reduction showing that $\text{\game{TQBF}}\leq_p \game{UPF}$. Therefore, \game{UPF} is \ps-hard, and, hence, \ps-complete, as required.
\end{proof}

\section{Undirected Impartial Geography}\label{s:ui1}

The game \game{UIR} is known in the literature as \game{Undirected Geography} or as \game{Undirected Vertex Geography,} and it is known to be in \textsc{p}~\cite{fsugeog}. The N-positions are those position where the token is on a vertex that is saturated by (i.e. included in) every maximum-size matching of the underlying graph, a condition which can be checked in polynomial time. Left wins by moving along an edge of some maximum-size matching. If the vertex with the token is not saturated by some maximum-size matching, then the position is a P-position, as every one of that vertex's neighbors must be saturated by every maximum-size matching excluding that vertex. We claim that, when played on a bipartite graph, the game \game{UIF} has the same N and P-positions, and Left can use the same strategy to win as in \game{UIR}. Our proof depends on the known strategy for \game{UIR}.

\begin{theorem}\label{thm:uif}
The game \game{UIF} is in \textsc{p} when the underlying graph is bipartite. In particular, a position is an N-position if and only if the vertex with the token is saturated by every maximum-size matching of the underlying graph, and Left can win from an N-position by moving regularly along an edge of some maximum-size matching.
\end{theorem}

We use the following lemma:
\begin{lemma}\label{lem:uif}
Let $G$ be a graph, and let $M_1$ and $M_2$ be matchings in $G$ each with the maximum possible number of edges. Each connected component of $M_1\cup M_2$ has an even number of edges and is either a path or a cycle.
\end{lemma}
\begin{proof}
The fact that each component is a path or cycle follows from the fact that $M_1$ and $M_2$ are matchings, meaning every vertex in $M_1\cup M_2$ has degree at most $2$. It immediately follows that any cycle must have an even number of edges, as otherwise two edges in the same matching would be incident. Now, suppose for a contradiction that a component $H$ of $M_1\cup M_2$ is a path with an odd number of edges. This means that one of the matchings (assume $M_1$ without loss of generality) has one more edge in $H$ than the other matching. Let $M=\pb{M_1\cap H}\cup\pb{M_2\cap\pb{G-H}}$. Note that $M$ is a matching of $G$ consisting of all of the edges of $M_2$ except in $H$, where it uses the edges of $M_1$. This means that $M$ contains more edges than $M_2$, contradicting the maximality of $M_2$. Therefore, $H$ has an even number of edges, as required.
\end{proof}

We now prove Theorem~\ref{thm:uif}.
\begin{proof}
Let $G$ be a bipartite graph, and let $v$ be the vertex of $G$ with the token on it. 
Suppose for a contradiction that this is a P-position of \game{UIR} but an N-position of \game{UIF}, meaning that a player can gain an advantage by moving irregularly. Furthermore, suppose $G$ has the fewest vertices among all bipartite 
graphs in positions where an irregular move can be beneficial. Note that $G$ has at least three vertices, since the position with a single vertex is a P-position in both \game{UIF} and \game{UIR}, and the positions with two vertices are either P-positions in both games (two isolated vertices) or N-positions in both games (a single edge).

Since the current position is a P-position for \game{UIR}, we know that $v$ is not saturated by some maximum-size matching $M_1$ of $G$. Note that every neighbor of $v$ must be saturated by $M_1$, as otherwise the matching's size could be increased by matching $v$ with an unmatched neighbor. The minimality property of $G$ means that the first move must involve deleting a vertex other than $v$ and that the rest of the game can be played as if it were \game{UIR} without affecting the outcome. Let $w$ be the neighbor of $v$ that Left should move to to win, and let $u$ be the neighbor of $w$ that Left should delete to win. The position after the move $\pb{v,w,u}$ is identical to the position after the move $\pb{u,w,u}$ if the token had started on $u$ instead of $v$. Since this winning move is a legal and winning move in \game{UIR}, it must be the case that $u$ is saturated by every maximum-size matching of $G$ and that $uw$ is an edge in at least one of these matchings. Let $M_2$ be such a matching. Note that $v$ must be matched to some vertex $x$ in $M_2$, as otherwise, $M_2-\st{uw}\cup\st{vw}$ would be a matching of $G$ of the same size as $M_2$, excluding $u$.

Let $H=M_1\cup M_2$. Let $H_v$ denote the connected component of $H$ containing $v$, and let $H_w$ denote the component of $H$ containing $w$. Since $uw\in M_2$, we have that $u\in H_w$ also. By Lemma~\ref{lem:uif}, each of $H_v$ and $H_w$ has an even number of edges and is either a path or a cycle. In particular, $H_v$ must be a path with $v$ at one end, since $M_1$ excludes $v$. There are two cases to consider:
\begin{description}
\item[$H_v\neq H_w$:] Let $M=\pb{M_1\cap H_v}\cup\pb{M_2\cap\pb{G-H_v}}$. This is a matching of $G$ consisting of all of the edges of $M_2$ except in $H_v$, where it uses the edges of $M_1$. This means that $M$ has the same number of edges as $M_2$, so it is a maximum-size matching that excludes $v$ and includes $uw$, a contradiction.
\item[$H_v=H_w$:] Since $uw\in M_2$, $H_v$ must consist of a path from $v$ to either $w$ or $u$ not through the other, then the edge $uw$, then at least one additional edge $wy$ (since $w$, a neighbor of $v$, is matched in $M_1$ and $u$ is saturated by every maximum-size matching, including $M_1$).

Suppose first that $w$ is encountered before $u$ when traveling from $v$ along $H_v$. Then, there is some edge $ut\in M_1$ incident to $u$. The path $P$ in $H_v$ from $v$ to $w$ has an even number of edges. Then, since $v$ and $w$ are neighbors in $G$, the cycle $C=P\cup\st{vw}$ has odd length, 
contradicting the fact that $G$ is bipartite.

Now, suppose that $u$ is encountered before $w$ when traveling from $v$ along $H_v$. Let $P$ denote the path in $H_v$ from $v$ to $u$, and let $Q$ denote the path in $H_v$ from $w$ to the other endpoint of $H_v$. Note that $P$ has an even number of edges and $Q$ has an odd number of edges, containing one more edge of $M_1$ than $M_2$. Let $M=\pb{M_2\cap P}\cup\pb{M_1\cap Q}\cup\pb{M_1\cap\pb{G-H_v}}$. This is a matching of the same size as $M_1$ that excludes $u$, a contradiction.
\end{description}
Therefore, $G$ cannot exist, and no advantage can be gained in \game{UIF} on 
a bipartite graph by moving irregularly, as required.
\end{proof}

Theorem~\ref{thm:uif} is false if the bipartite condition is dropped. Only one sub-case of the proof depends on the graph being bipartite, but it is possible to construct counterexamples based on that case. The smallest non-bipartite counterexample, with five vertices and five edges, is depicted in Figure~\ref{fig:uifcounter}. If the token begins on the vertex $v$ in the figure, the resulting position is a P-position for \game{UIR} and an N-position for \game{UIF}. This graph has four maximum size matchings (two edges each). Three of them consist of $ut$ and an edge of the triangle; the fourth matching is $\st{uw, vx}$. Note that $u$ is saturated by every one of these matchings, but $v$ is excluded by the matching $\st{ut, xw}$. Since $v$ is saturated by every maximum-size matching that includes $uw$, we have that the depicted position is a P-position in \game{UIR} and an $N$-position in \game{UIF} via the move $\pb{v,w,u}$.

\begin{figure}
\begin{center}
\begin{tikzpicture}[thick, node distance=1.5cm, circ/.style={draw, circle, minimum size=35pt}]

\node[circ] at (0, 0) (v) {$v, T$};
\node[circ] at (0, -2cm) (x) {$x$};
\node[circ] at (1.5cm, -1cm) (w) {$w$};
\node[circ] at (3.5cm, -1cm) (u) {$u$};
\node[circ] at (5.5cm, -1cm) (z) {$t$};

\draw (v) to (x);
\draw (v) to (w);
\draw (x) to (w);
\draw (w) to (u);
\draw (u) to (z);
\end{tikzpicture}
\end{center}
\caption{Smallest counterexample to Theorem~\ref{thm:uif} on a non-bipartite graph.}
\label{fig:uifcounter}
\end{figure}
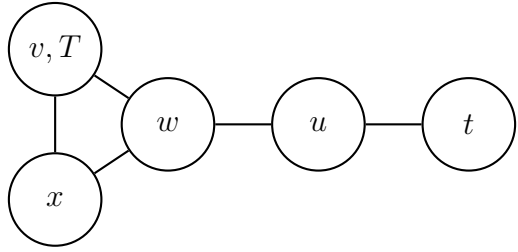

\section{Undirected Impartial Geography with Stacking}\label{s:uik}
As discussed in Section~\ref{s:ui1}, the game \sgame{UIR}{1} is in \textsc{p}. As a consequence of Proposition~\ref{prop:1vs2}, the game \sgame{UIR}{2} is also known to be in \textsc{p}. We now show that the game \sgame{UIR}{4} is \ps-complete, even on bipartite graphs with maximum degree at most $3$. Unlike our other reductions, we reduce from \game{Geography} rather than from \game{TQBF}.
\begin{theorem}\label{thm:uir4}
The game \sgame{UIR}{4} is \ps-complete, even when restricted to bipartite graphs with maximum degree at most $3$.
\end{theorem}

We use the following lemma:
\begin{lemma}\label{lem:20133}
Consider the game of \sgame{UIR}{4} played on a path of five vertices $v_1$, $v_2$, $v_3$, $v_4$, and $v_5$ with initial heights $2$, $4$, $3$, $1$, and $1$ respectively with the token starting on $v_1$ (see Figure~\ref{fig:20133}). Modify the game so that it is a draw if the token reaches $v_5$ with Left to move. Then, under optimal play by both players, this game is a draw, and the ending heights are one of the following sequences (listed from $v_1$ through $v_5$): $[1,2,1,0,1]$, $[1,1,0,0,1]$, or $[0,0,0,0,1]$.
\end{lemma}
\begin{figure}
\begin{center}
\begin{tikzpicture}[thick, node distance=1.5cm, circ/.style={draw, circle, minimum size=30pt}]

\node[circ] at (0, 0) (v1) {$2T$};
\node[circ] at (2cm, 0) (v2) {$4$};
\node[circ] at (4cm, 0) (v3) {$3$};
\node[circ] at (6cm, 0) (v4) {$1$};
\node[circ] at (8cm, 0) (v5) {$1$};

\node at (0, -1cm) (v1l) {$v_1$};
\node at (2cm, -1cm) (v2l) {$v_2$};
\node at (4cm, -1cm) (v3l) {$v_3$};
\node at (6cm, -1cm) (v4l) {$v_4$};
\node at (8cm, -1cm) (v5l) {$v_5$};

\draw (v1) to (v2);
\draw (v2) to (v3);
\draw (v3) to (v4);
\draw (v4) to (v5);
\end{tikzpicture}
\end{center}
\caption{The game under consideration in Lemma~\ref{lem:20133}}
\label{fig:20133}
\end{figure}
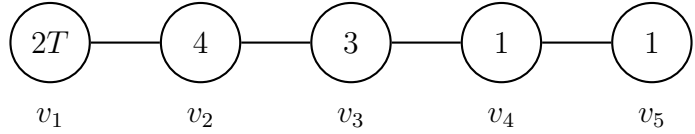
\begin{proof}
Note first that the underlying graph is bipartite, so it is always Left's turn if the token's vertex has odd index and Right's turn if the index is even. So, if the token reaches $v_5$, it is always Left's turn and hence a draw.

This game is small enough that it can be solved by brute force. The first four moves are forced: to $v_2$, to $v_3$, to $v_2$, and to $v_3$. The resulting position has heights $[1,2,2,1,1]$ and the token on $v_3$ with Left to move. She can move to either $v_2$ or $v_4$, resulting in the heights $[1,2,1,1,1]$ in either case.

Suppose first that Left moves to $v_2$ on the fifth move. Then, Right can move to either $v_1$ or $v_3$, making all heights equal $1$ in the process. If he moves to $v_1$, all remaining moves are forced, and the token ends up on $v_5$ with heights $[0,0,0,0,1]$. If Right moves to $v_3$, left loses immediately if she moves to $v_2$, so she moves to $v_4$, forcing Right to move to $v_5$, drawing the game with heights $[1,1,0,0,1]$.

Now, suppose that Left moves to $v_4$ on the fifth move. Right can immediately draw the game by moving to $v_5$, leading to heights $[1,2,1,0,1]$. Right's other possible move is to $v_3$, also leaving these heights. This forces the rest of the game to proceed by Left moving to $v_2$, Right to $v_1$,  and Left back to $v_2$, at which point Right loses as the heights are $[0,1,0,0,1]$. So, Right should draw the game on the sixth move rather than play the losing move to $v_3$.

Figure~\ref{fig:20133tree} depicts the full game tree. A position is labeled $(L)$ if Left wins, $(R)$ if Right wins, or $(D)$ if the position is a draw.
\end{proof}

\begin{figure}
\begin{center}
\input{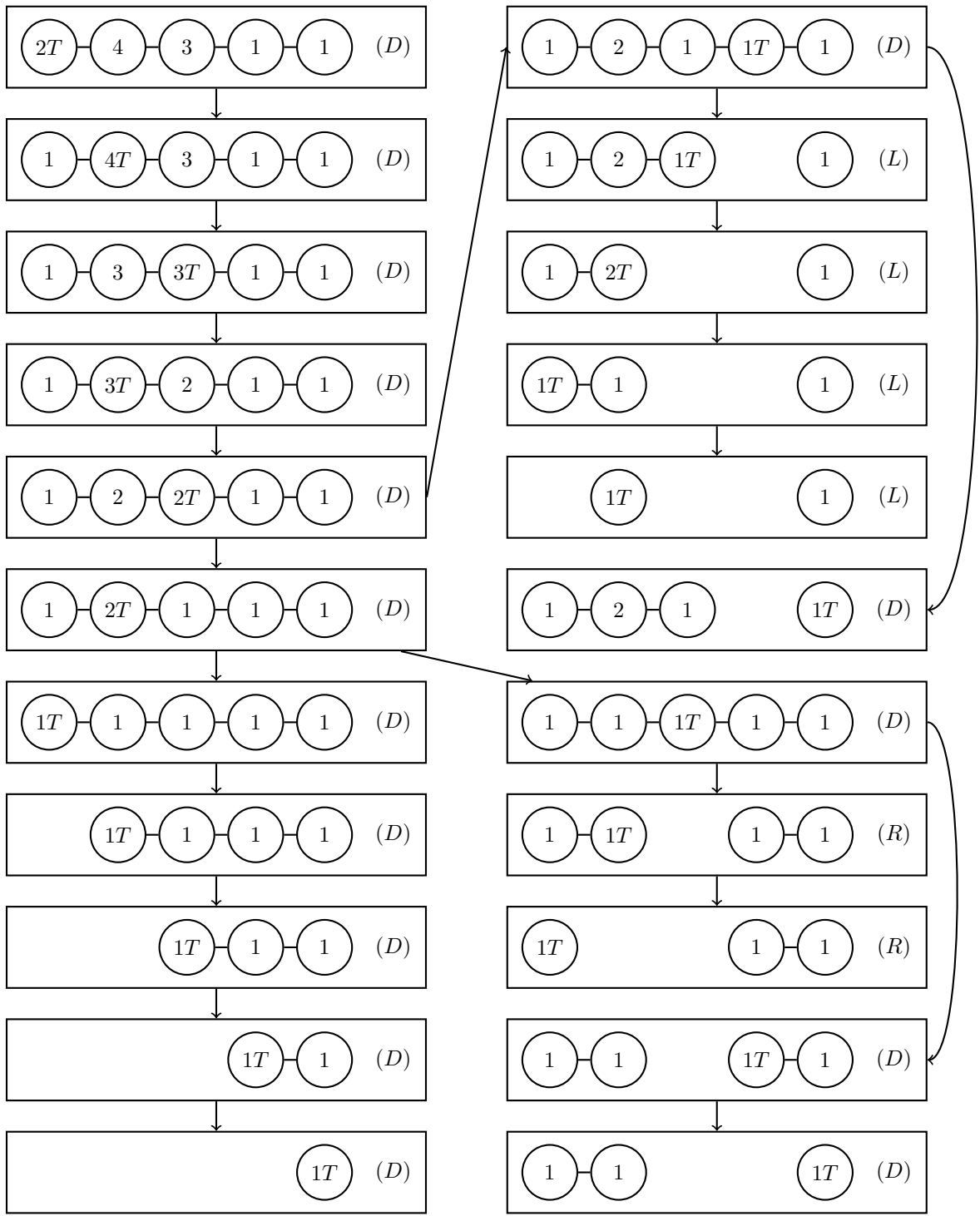}
\end{center}
\caption{The game tree analyzed in the proof of Lemma~\ref{lem:20133}}
\label{fig:20133tree}
\end{figure}

We now prove Theorem~\ref{thm:uir4}.

\begin{proof}
We proceed via a reduction from \game{Geography} (\sgame{DIR}{1}) to \sgame{UIR}{4}. Consider an instance of \game{Geography} on underlying digraph $G$ with the token on vertex $t$. We construct an instance of \sgame{UIR}{4} on graph $G'$ as follows. 
\begin{itemize}
\item Each vertex $v$ in $G$ corresponds to five vertices $v_1$, $v_2$, $v_3$, $v_4$, and $v_5$ in $G'$. The initial heights of these vertices are $2$, $4$, $3$, $1$, and $1$ respectively.
\item The token begins on vertex $t_1$.
\item For each arc $vw$ in $G$, there is an edge $v_5w_1$ in $G'$.
\end{itemize}
See Figure~\ref{fig:uir4} for an example of this construction.

\begin{figure}
\begin{center}
\begin{tikzpicture}[thick, node distance=1.5cm, circ/.style={draw, circle, minimum size=25pt}, scale=0.9]

\node[circ] at (0, 0) (v) {$T$};
\node[circ] at (0, -3) (w) {};
\node[circ] at (0, -6) (x) {};
\node[circ] at (2, -1.5) (y) {};
\node[circ] at (2, -4.5) (z) {};

\draw[->] (v) to (w);
\draw[->] (v) to (y);
\draw[->] (w) to (x);
\draw[->] (x) to (z);
\draw[->] (y) to (w);
\draw[->] (y) to (z);
\draw[->] (z) to (w);
\draw[->] (z) to [out=50, in=20] (v);

\node at (3.5, -3) (arrow) {$\Longrightarrow$};

\node[circ] at (5, 0) (v1) {$2T$};
\node[circ] at (6.5, 0) (v2) {$4$};
\node[circ] at (8, 0) (v3) {$3$};
\node[circ] at (9.5, 0) (v4) {$1$};
\node[circ] at (11, 0) (v5) {$1$};

\node[circ] at (5, -3) (w1) {$2$};
\node[circ] at (6.5, -3) (w2) {$4$};
\node[circ] at (8, -3) (w3) {$3$};
\node[circ] at (9.5, -3) (w4) {$1$};
\node[circ] at (11, -3) (w5) {$1$};

\node[circ] at (5, -6) (x1) {$2$};
\node[circ] at (6.5, -6) (x2) {$4$};
\node[circ] at (8, -6) (x3) {$3$};
\node[circ] at (9.5, -6) (x4) {$1$};
\node[circ] at (11, -6) (x5) {$1$};

\node[circ] at (8, -1.5) (y1) {$2$};
\node[circ] at (9.5, -1.5) (y2) {$4$};
\node[circ] at (11, -1.5) (y3) {$3$};
\node[circ] at (12.5, -1.5) (y4) {$1$};
\node[circ] at (14, -1.5) (y5) {$1$};

\node[circ] at (8, -4.5) (z1) {$2$};
\node[circ] at (9.5, -4.5) (z2) {$4$};
\node[circ] at (11, -4.5) (z3) {$3$};
\node[circ] at (12.5, -4.5) (z4) {$1$};
\node[circ] at (14, -4.5) (z5) {$1$};

\draw (v1) to (v2);
\draw (v2) to (v3);
\draw (v3) to (v4);
\draw (v4) to (v5);

\draw (w1) to (w2);
\draw (w2) to (w3);
\draw (w3) to (w4);
\draw (w4) to (w5);

\draw (x1) to (x2);
\draw (x2) to (x3);
\draw (x3) to (x4);
\draw (x4) to (x5);

\draw (y1) to (y2);
\draw (y2) to (y3);
\draw (y3) to (y4);
\draw (y4) to (y5);

\draw (z1) to (z2);
\draw (z2) to (z3);
\draw (z3) to (z4);
\draw (z4) to (z5);

\draw (v5) to [out=210, in=80] (w1);
\draw (v5) to [out=245, in=35, looseness=0] (y1);
\draw (w5) to [out=210, in=80, looseness=0.7] (x1);
\draw (x5) to (z1);
\draw (y5) to [out=210, in=30] (w1);
\draw (y5) to [out=260, in=30, looseness=0.7] (z1);
\draw (z5) to [out=150, in=-30, looseness=0.45] (w1);
\draw (z5) to [out=70, in=30, looseness=1.6] (v1);
\end{tikzpicture}
\end{center}
\caption{\sgame{UIR}{4} construction from a given \game{Geography} instance}
\label{fig:uir4}
\end{figure}
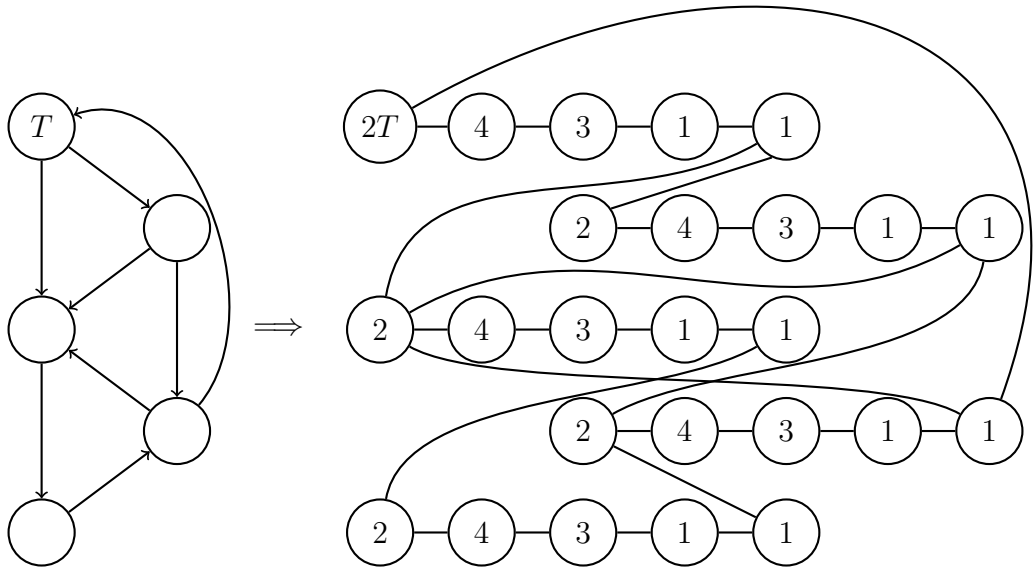

Refer to the set of five vertices in $G'$ obtained from a vertex of $G$ as an \emph{meta-vertex.} Say that a meta-vertex is \emph{pristine} if no move has been made from one vertex in the meta-vertex to another. We say that play of this game proceeds \emph{properly} if the play obeys the following guiding principles for as long as possible.
\begin{itemize}
\item If the token is not on a vertex of the form $v_5$, then the token is moved within the current meta-vertex in such a way that each player avoids losing.
\item If the token is on a vertex of the form $v_5$, then the token is moved to a vertex of the form $w_1$ within a pristine meta-vertex.
\end{itemize}

We have the following claims, which would complete our reduction:
\begin{itemize}
\item Any player making an improper move loses.
\item If the game progresses properly for as long as possible, then the winner of the \game{Geography} position is the same as the winner of the corresponding \sgame{UIR}{4} position.
\end{itemize}

We show first that any improper move leads to loss. First, we show that neither player has an incentive to make a move from a vertex $v_1$ to a vertex $w_5$. Without loss of generality, suppose that Left makes such a move as the first improper move of the game. Since the vertex moved to has height $1$, it cannot have been moved to previously. In other words, the move is to a pristine meta-vertex. Then, Right can move from $w_5$ to $w_4$, which results in the deletion of $w_5$. This forces Left's next move to be from $w_4$ to $w_3$, deleting $w_4$. Then, Right moves from $w_3$ to $w_2$ decreasing the height of $w_3$ to $2$ and trapping Left on $w_2$, which has height $4$. So, a player making an improper move of this type loses.

Now, consider the other possible type of improper move: a move from a vertex $v_5$ to a vertex $w_1$ of a non-pristine meta-vertex. Without loss of generality, suppose that Left makes such a move as the first improper move of the game.  We claim that the heights of $w_1$ through $w_5$ must be $\bk{1,2,1,0,1}$ or $\bk{1,1,0,0,1}$. Since this meta-vertex is not pristine, play must have properly proceeded through it previously. Since we know that all moves previously made involving that meta-vertex started with the token on $w_1$ and remained within the vertex until the token reached $w_5$ and subsequently departed, play must have proceeded as if the game in Lemma~\ref{lem:20133} were being played within that vertex. So, the final heights must be as described by Lemma~\ref{lem:20133}. Since the move from $v_5$ to $w_1$ is possible, it cannot be that $w_1$ has height $0$, so the heights must be $\bk{1,2,1,0,1}$ or $\bk{1,1,0,0,1}$, as required.

If the heights are $\bk{1,1,0,0,1}$, Right immediately wins by moving from $w_1$ to $w_2$. If the heights are $\bk{1,2,1,0,1}$, Right can move to $v_2$, then Left to $v_3$, then Right back to $v_2$, again trapping Left. So, a player making an improper move of this type also loses.

Now, we know that both players are incentivized to move properly whenever possible. By Lemma~\ref{lem:20133}, if $G$ consists of a single vertex, Right wins both the \game{Geography} position and the \sgame{UIR}{4} position, as a draw in the lemma is a loss for Left in this situation. Now, suppose $G$ has at least two vertices, and suppose inductively that the winners of the \game{Geography} position and the corresponding \sgame{UIR}{4} position are the same for any \game{Geography} position with fewer vertices than $G$. Note that, as long as the token remains in the meta-vertex of $G'$ corresponding to $v$ in $G$, it is always Left's move when the token is on $v_1$ or $v_5$ because each meta-vertex is bipartite. Since Left is never incentivized to move improperly, play through the first meta-vertex proceeds according to Lemma~\ref{lem:20133}, as the best Left can hope for from the first meta-vertex is to progress to another meta-vertex. (In other words, Left cannot win within the first meta-vertex.) Consider the \sgame{UIR}{4} position reached once Left makes a move from $v_5$ to some vertex $w_1$. This position is the same as the \sgame{UIR}{4} position corresponding to the \game{Geography} position resulting from Left moving from $v$ to $w$ in $G$, except that vertices $v_1$, $v_2$, and $v_3$ may still exist in $G'$. Moving to $v_1$ later would be an improper move, and hence losing, so the current position has the same winner as the one in which the entire meta-vertex corresponding to $v$ has been deleted. By our minimality property, the winner of the current position is the same as the winner of the \game{Geography} position after Left moves from $v$ to $w$ in $G$. Since Left can move to the meta-vertex for $w$ in $G'$ if and only if she can move to $w$ in $G$, the initial \sgame{UIR}{4} position has the same winner as the initial \game{Geography} position, as required.

This completes the reduction showing that $\text{\game{Geography}}\leq_p \text{\sgame{UIR}{4}}$. Therefore, \sgame{UIR}{4} is \ps-hard, and, hence, \ps-complete, as required.

Finally, we know that \game{Geography} is still \ps-complete when restricted to bipartite graphs of maximum in/out degree at most $2$ and maximum total degree at most $3$. We claim that if $G$ has these properties, then $G'$ is bipartite with maximum degree at most $3$. First, suppose that $G$ has bipartition $\pb{A,B}$. We construct a bipartition $\pb{A',B'}$ of $G'$ as follows: If $v\in A$, then $v_1,v_3,v_5\in A'$ and $v_2,v_4\in B'$, and vice versa if $v\in B$. Second, each incoming arc to a vertex $v$ of $G$ translates to an edge of $G'$ incident to $v_1$, and each outgoing arc translates to an edge of $G'$ incident to $v_5$. If the maximum in/out degree of $G$ is at most $2$, then there are at most two such edges incident to each of these. Since each of $v_1$ and $v_5$ is adjacent to one additional vertex ($v_2$ or $v_4$), each of these vertices has degree at most $3$. This means that $G'$ has maximum degree at most $3$, since vertices with indices $2$, $3$, or $4$ have degree $2$.
\end{proof}

\section{Summary and Open Problems}\label{s:concl}

Prior to this work, the complexities of only three of the games \sgame{XYZ}{1} were known (\game{DIR}, \game{DPR}, and \game{UIR}). We resolve the complexities of four more of these games (\game{UPR}, \game{DIF}, \game{DPF}, and \game{UPF}) and a restricted case of the fifth (\game{UIF}). An obvious open problem is to fully resolve the complexity of \game{UIF}. Furthermore, the complexity of each game except \game{UPF} is known for bipartite graphs. Also, in Section~\ref{s:uik}, we determine that \sgame{UIR}{4} is \ps-complete. Since we also know that \sgame{UIR}{2} is in \textsc{p}, this leaves open the complexity of the game \sgame{UIR}{3}.

As mentioned in the introduction, the \game{Geography} variants studied in this paper are motivated by the board game \textit{Santorini}. Future developments could continue to bridge the gap between \game{Geography} and \textit{Santorini}. Here are several more variant rules that can be considered:
\begin{itemize}
\item The underlying graph can be restricted to come from some additional families closed under taking induced subgraphs. We have seen bipartite, planarity, and degree restrictions.
\emph{Santorini} is played on a grid (including diagonal adjacencies), so another restriction worth studying is grids and their induced subgraphs.
\item \textit{Santorini} has a winning condition other than the opponent having no legal move. In the language of our  \game{Geography} variants with stacked vertices, this condition is equivalent to ending one's move on a vertex of height $1$. Introducing alternative winning conditions is an avenue for further study.
\item In \textit{Santorini}, each player has two tokens, and the players choose each token's starting position in turn order prior to moving them. Adding additional tokens or token-placement moves could theoretically affect the complexities of games. Some work in this direction has been undertaken for other games in the \game{Geography} family~\cite{fsgeog}.
\item In our proof of Theorem~\ref{thm:uir4}, we assign each vertex a starting height. The position we construct could not have been obtained from an initial position with all vertices having maximum height, which may be a desirable condition to consider.
\item Throughout our work, we treat the maximum height as a fixed parameter, as the length of the game has a polynomial dependence on the maximum height $k$, and hence an exponential dependence on $\log k$, the number of bits needed to represent $k$. This suggests that if maximum height is unrestricted, the game is unlikely to be in \ps. The question then is, what is the complexity of this unrestricted game?
\end{itemize}

\bibliography{bibliography.bib}

\begin{bibdiv}
\begin{biblist}

\bib{winways}{book}{
      author={Berlekamp, Elwyn~R},
      author={Conway, John~Horton},
      author={Guy, Richard~K},
       title={Winning ways for your mathematical plays},
   publisher={A. K. Peters},
        date={2001},
}

\bib{bodlaender1993complexity}{article}{
      author={Bodlaender, Hans~L},
       title={Complexity of path-forming games},
        date={1993},
     journal={Theoret. Comput. Sci.},
      volume={110},
      number={1},
       pages={215\ndash 245},
}

\bib{nim}{article}{
      author={Bouton, Charles~L},
       title={Nim, a game with a complete mathematical theory},
        date={1901},
     journal={Ann. of Math.},
      volume={3},
      number={1/4},
       pages={35\ndash 39},
}

\bib{darmann2014shortest}{inproceedings}{
      author={Darmann, Andreas},
      author={Pferschy, Ulrich},
      author={Schauer, Joachim},
       title={The shortest path game: Complexity and algorithms},
organization={Springer},
        date={2014},
   booktitle={{I}{F}{I}{P} {I}nternational {C}onference on {T}heoretical
  {C}omputer {S}cience},
       pages={39\ndash 53},
}

\bib{fsgeog}{article}{
      author={Fraenkel, Aviezri},
      author={Simonson, Shai},
       title={Geography},
        date={199303},
     journal={Theor. Comput. Sci.},
      volume={110},
       pages={197\ndash 214},
}

\bib{chess}{inproceedings}{
      author={Fraenkel, Aviezri~S},
      author={Lichtenstein, David},
       title={Computing a perfect strategy for n$\times$ n chess requires time
  exponential in n},
organization={Springer},
        date={1981},
   booktitle={International colloquium on automata, languages, and
  programming},
       pages={278\ndash 293},
}

\bib{fsugeog}{article}{
      author={Fraenkel, Aviezri~S},
      author={Scheinerman, Edward~R},
      author={Ullman, Daniel},
       title={Undirected edge geography},
        date={1993},
     journal={Theor. Comput. Sci.},
      volume={112},
      number={2},
       pages={371\ndash 381},
}

\bib{gareyjohnson}{book}{
      author={Garey, Michael~R},
      author={Johnson, David~S},
       title={Computers and intractability},
   publisher={Freeman San Francisco},
        date={1979},
      volume={174},
}

\bib{santorini}{misc}{
      author={Hamilton, Gordon},
       title={Santorini},
   publisher={Roxley Games},
        date={2017},
}

\bib{karp}{incollection}{
      author={Karp, Richard~M.},
       title={Reducibility among combinatorial problems},
        date={1972},
   booktitle={Complexity of computer computations},
   publisher={Springer},
       pages={85\ndash 103},
}

\bib{geogreduct}{article}{
      author={Lichtenstein, David},
      author={Sipser, Michael},
       title={G{O} is polynomial-space hard},
        date={1980},
     journal={J. ACM},
      volume={27},
      number={2},
       pages={393\ndash 401},
}

\bib{monti2018variants}{article}{
      author={Monti, Angelo},
      author={Sinaimeri, Blerina},
       title={On variants of vertex geography on undirected graphs},
        date={2018},
     journal={Discrete Appl. Math.},
      volume={251},
       pages={268\ndash 275},
}

\bib{schaefergeog}{article}{
      author={Schaefer, Thomas~J},
       title={On the complexity of some two-person perfect-information games},
        date={1978},
     journal={J. Comput. System Sci.},
      volume={16},
      number={2},
       pages={185\ndash 225},
}

\end{biblist}
\end{bibdiv}

\end{document}